\documentclass[a4paper,12pt]{article}
\usepackage[toc,page]{appendix}
\usepackage[numbers]{natbib}
\usepackage{rotating}
\usepackage{amsfonts}
\usepackage{amssymb}
\usepackage{caption}
\usepackage{amsmath}
\usepackage{amsthm}
\usepackage{url}
\usepackage{accents}
\usepackage{graphicx}
\usepackage[math]{}
\usepackage[toc,page]{appendix}
\usepackage[]{graphicx}
\usepackage{caption}
\usepackage{subcaption}

\newcommand{\go}{\gamma_0}
\newcommand{\eo}{\eta_0}
\newcommand{\dd}{\delta}
\newcommand{\goL}{\gamma_0^L}
\newcommand{\gol}{\gamma_0^l}
\newcommand{\goG}{\gamma_0^G}
\newcommand{\gog}{\gamma_0^g}
\newcommand{\eog}{\eta_0^g}
\newcommand{\eoL}{\eta_0^L}
\newcommand{\eoG}{\eta_0^G}
\newcommand{\eol}{\eta_0^l}
\newcommand{\Go}{\Gamma_0}
\newcommand{\Eo}{\Xi_0}
\newcommand{\nubfuno}{|f_1|_{1-2\alpha}}
\newcommand{\nubguno}{|g_1-\mathcal{G}|_{1-2\alpha}}
\newcommand{\nubfunos}{|f_1|_{3-2\alpha}}
\newcommand{\nubgunos}{|g_1-\mathcal{G}|_{3-2\alpha}}
\newcommand{\ndbfo}{||X_{f_0}||_1}
\newcommand{\ndbgo}{||X_{g_0}-X_{\mathcal{G}}||_1}
\newcommand{\ndbgc}{||X_{\mathcal{G}}||_1}
\newcommand{\ntbfo}{|||X_{f_0}|||_1}
\newcommand{\ntbgo}{|||X_{g_0}-X_{\mathcal{G}}|||_1}
\newcommand{\ndbfuno}{||X_{f_1}||_{1-2\alpha}}
\newcommand{\ndbguno}{||X_{g_1}-X_{\mathcal{G}}||_{1-2\alpha}}
\newcommand{\ntbfuno}{|||X_{f_1}|||_{1-2\alpha}}
\newcommand{\ntbguno}{|||X_{g_1}-X_{\mathcal{G}}|||_{1-2\alpha}}
\newcommand{\ndbfunos}{||X_{f_1}^{j}||_{3-2\alpha}}
\newcommand{\ndbgunos}{||X_{g_1}^{j}-X_{\mathcal{G}}^{j}||_{3-2\alpha}}
\newcommand{\auno}{\alpha}
\newcommand{\dom}[1]{\mathcal{D}_{#1}}

\newcommand{\flphit}{\Lambda_{\phi_1}^t}
\newcommand{\flphits}{\Lambda_{\phi_1}^{t^*}}

\newcommand{\integ}[4]{\int_{#1}^{#2} #3 d#4 }
\newcommand{\normsup}[2]{|#1|_{#2}}
\newcommand{\ch}[1]{X_{#1}}
\newcommand{\nub}[2]{\left| #1 \right|_{#2}}
\newcommand{\ndb}[2]{\left|\left| #1 \right|\right|_{#2}}
\newcommand{\ntb}[2]{\left|\left|\left| #1 \right|\right|\right|_{#2}}
\newcommand{\jj}{\mathcal{J}}
\newcommand{\der}[2]{\frac{\partial #1}{\partial #2}}
\newcommand{\lie}[2]{\left[#1,#2\right]}
\newcommand{\ti}{\tilde{I}}
\newcommand{\tth}{\tilde{\vartheta}}
\newcommand{\tx}{\tilde{x}}
\newcommand{\ty}{\tilde{y}}
\newcommand{\tr}{\tilde{R}}
\newcommand{\ndbs}[3]{\left|\left| #1^{#3} \right|\right|_{#2}}

\title{Sharp Nekhoroshev estimates for the three-body problem around periodic orbits}
\author{Santiago Barbieri 
\footnote{Dipartimento di Matematica {\itshape Tullio Levi Civita}, via Trieste 63, 35131 Padova, Italy - Corresponding author: santiago.barbieri@math.unipd.it}
\footnote{This work has been developed under the auspices of the European Research Council in the framework of the H2020-ERC Starting Grant 2015 project 677793: {\itshape Stable and Chaotic Motions in the Planetary Problem.} The author warmly acknowledges the ERC for making this possible.
} 
\hspace{25pt}  Laurent Niederman
\footnote{Universit\'{e} Paris XI B\^{a}t. 425, 91405 Orsay Cedex, France \& \newline IMCCE-Observatoire de Paris, 77 avenue Denfert-Rochereau, 75014 Paris Cedex, France}
\\
\date{}
}
\begin{document}
\maketitle

\section*{Abstract}
We construct a Nekhoroshev-like result of stability with sharp constants for the planar three body problem, both in the planetary and in the restricted circular case, by using the periodic averaging technique. Our constructions can be generalized to any near-integrable hamiltonian system whose unperturbed hamiltonian is quasi-convex. The dependence of the constants on the analyticity widths of the complex hamiltonian is carefully taken into account. This allows for a deep analytical understanding of the limits of such techniques in insuring Nekhoroshev stability for high magnitudes of the perturbation and suggests hints on how to overcome such obstructions in some cases. Finally, two examples with concrete values are considered, one for the planetary case and one for the restricted one.  
\section{Introduction}
It is well known since the end of the 19th century that the problem of $n$ point masses mutually interacting by the sole gravitational force is non-integrable for $n\geq 3$ (see \cite{Chenciner_2012} for a detailed historical overview on this subject). Coming to more recent times, the birth of KAM theory in the mid-twentieth century led to new mathematical efforts in order to establish whether quasi-periodic motions persisted in the $n$-body problem for suitable perturbative parameters. In particular, important results of stability based on KAM theory were achieved in \cite{Arnold_1963} for the planar three-body problem, in \cite{Robutel_1995} for the spatial case and in \cite{Chierchia_Pinzari_2011}, \cite{Fejoz_2004}, \cite{Pinzari_2009} for the general $n$-body problem. Numerical studies (see e.g. \cite{Laskar_1994}) show that the motion of the outer Solar System stays stable for timescales which exceed the lifetime of the universe, so that purely analytical investigations on the stability of the major planets make sense. Moreover, the direct application of KAM theorems to the $n$-body problem, with $n\geq 3$, usually leads to pessimistic estimates on the maximal size that the perturbation can reach in order for such results to hold (see \cite{Castan_2017} for a recent discussion on this issue). On the other hand, good estimates can be obtained when considering the invariance of particular tori under the dynamics of a suitably truncated perturbation, as it is done in \cite{Celletti_Chierchia_2007}. 
\newline
Another possibility is to apply the less-demanding Nekhoroshev theorem to such problem in order to insure that the perturbed system stay close to the integrable one over exponentially long times. Indeed, though leading to a weaker, non-perpetual form of stability, Nekhoroshev theorem requires less strict conditions and yields bounds on the perturbative parameters which are closer to realistic ones (see e.g. \cite{Niederman_1996}). Moreover, such result holds on open sets. Two different proofs of such statement exist: the original one by Nekhoroshev \cite{Nekhoroshev_1977} and the one developed by Lochak in \cite{Lochak_1992}. The first approach insures a slow rate of diffusion of the action variables over exponentially long times under the generic assumption that the unperturbed system satisfies a condition known as {\itshape steepness}. Such result has been improved in \cite{Benettin_Galgani_Giorgilli_1985} and in \cite{Poschel_1993} for the convex case and in \cite{Guzzo_Chierchia_Benettin_2016} for the original {\itshape steep} case. The second proof works under the hypothesis that the unperturbed hamiltonian is quasi-convex and exploits such geometrical property in order to insure exponential times of stability in the neighborhood of periodic orbits of the unperturbed system. A global result of stability is obtained once one covers the entire phase space with such neighborhoods with the help of Dirichlet's approximation theorem. Improvements in this second approach can be found in \cite{Bounemoura_Niederman_2012} and \cite{Lochak_Neishtadt_Niederman_1994}. A brief overview on both proofs can be found in \cite{Guzzo_2015} and \cite{Niederman_2009}. 
\newline
As for the applications to celestial mechanics, Niederman carefully derived in \cite{Niederman_1996} estimates of stability over exponentially long times for the three body planetary problem around a periodic torus. In the case of the $5:2$ resonance, stability holds for a time comparable with the age of the Solar System if the ratio for the mass of the greater planet on the Sun mass does not exceed $10^{-13}$ (the real value is actually $10^{-3}$ in the Solar System). On the other hand, numerical-assisted studies on Nekhoroshev stability, with realistic magnitudes for the perturbation, have been achieved by Giorgilli, Locatelli and Sansottera in \cite{Giorgilli_Locatelli_Sansottera_2009} and \cite{Giorgilli_Locatelli_Sansottera_2017} for a suitably truncated three or even four body hamiltonian in the neighborhood of an invariant torus. An application leading to a remarkably good upper bound on the perturbative parameter ($\varepsilon<10^{-6}$) in the non-resonant restricted, circular, planar case has also been considered by Celletti and Ferrara in \cite{Celletti_Ferrara_1996}.  Finally, an interesting discussion on the threshold on the magnitude of the perturbation for Nekhoroshev theorem to hold can be found in \cite{Bounemoura_2012}.
\newline
\newline
With respect to the present work, we intend to reach multiple goals which can be summarized as follows:
\begin{enumerate}
\item The first aim consists in obtaining a Nekhoroshev-like stability result with sharp constants for the planar three-body problem with the help of refined estimates on hamiltonian vector fields. Actually, our proof can be developed for any near-integrable hamiltonian system.
\item Secondly, we want to compare such result on the planetary three-body problem to those of Niederman in \cite{Niederman_1996} and see if sharp estimates lead to improvements in the time of stability and in the maximal allowed size for the perturbation.
\item With the help of the previous results, we want to be able to understand which are the analytical obstacles in this reasoning that prevent one from reaching physical values for the perturbation in the planetary case and conjecture how to overcome them in some cases. 
\item Finally, we will consider an application of the previous results to the restricted, circular three-body problem as modeled in \cite{Celletti_Chierchia_2007} and \cite{Celletti_Ferrara_1996}; as in the previous case, this will allow for a deeper understanding of the limits of the theory we make use of and, moreover, will open the possibility for reaching realistic values in the perturbative parameters once suitably powerful numerical tools are implemented. 
\end{enumerate}
The authors conjecture that the deadlocks encountered by the theory in such framework are general and can be considered as fundamental in any application of Nekhoroshev theory to finite-dimensional systems close to periodic integrable orbits. 
\newline
\newline
The paper is structured as follows: in paragraph \ref{Notations} we introduce notations and in section \ref{tbpp} the Nekhoroshev stability of the plane, planetary three-body problem is investigated with sharp techniques leading to sharp constants. Chapter \ref{restricted_tbp} is devoted to an application of our previous result to the restricted, circular, planar three-body problem, whereas section \ref{num_comp} contains applications to concrete examples. 
\section{Notations}\label{Notations}
In this section, we give some definitions that will be used throughout this work. 
\newline
In order for the calculations which will appear in the next chapters to be carried on, one must consider the following sets.
\newtheorem{sets}{Definition}
\begin{sets}
We define the real balls 
\begin{align}
\begin{split}
S_{I_0}(\rho) & := \{I\in\mathbb{R}: |I-I_0|<\rho\}\ ,\\
B_{x_0,y_0}(\xi) & := \{(x,y)\in\mathbb{R}^2: (x-x_0)^2+(y-y_0)^2<\xi^2\}
\end{split}
\end{align}
and the complex domain
\begin{align}\label{dominiio}
\begin{split}
\dom{\rho,r,s,\xi,u}:=\{& (I_1,I_2,\vartheta_1,\vartheta_2,x_1,x_2,y_1,y_2) \in\mathbb{C}^8:\\
& \exists \ I_j^*\in S_{0}(\rho) \text{ such that } \left|I_j-I_j^*\right|<r,\ j\in\{1,2\}\ ,\\
& \Re e(\vartheta_1,\vartheta_2)\in\mathbb{T}^2\ ,\ \ \left|\Im m(\vartheta_j)\right|<s,\ j\in\{1,2\}\ ,\\
& \exists (x_j^*,y_j^*)\in B_{0,0}(\xi) \text{ such that }\\
& \left|x_j-x_j^*\right|<u,\ \left|y_j-y_j^*\right|<u, \ j\in\{1,2\}
\}\ .
\end{split}
\end{align}
For the sake of simplicity, since the quantities we will deal with in the sequel are just $r,s$ and $u$, the last set will often be denoted by making use of some shorthand notations, namely
\begin{align}\label{shorthands}
\begin{split}
\dom{r,s,u} & := \mathcal{D}_{\rho,r,s,\xi,u}\ ,\\
\dom{\alpha-\beta} & :=\dom{r(\alpha-\beta),s(\alpha-\beta),u(\alpha-\beta)}\ ,\ \ 0\leq\beta\leq\alpha\\
\dom{\alpha,\beta} & := \dom{\alpha r,\alpha s,\beta u}\ .
\end{split}
\end{align}
\end{sets}
\
\newline
Now, let $F$ be a continuous scalar function of many complex variables bounded in an open domain $\mathcal{A}$, i.e.
$$
F: \mathcal{A}\subset\mathbb{C}^n \longrightarrow \mathbb{C}\ ,\ \ 
z\longmapsto F(z)\ ,\ \ \sup_{z\in\mathcal{A}}|F(z)|<+\infty\ .
$$
\newtheorem{nubar}[sets]{Definition}
\begin{nubar}
We denote the $\sup$-norm of $F$ with
$$
\nub{F}{\mathcal{A}} := \sup_{z\in\mathcal{A}}|F(z)|\ .
$$ 
\end{nubar}
\
\newline
A natural extension of this definition applies when considering a continuous vector-valued function 
$$
v:\mathcal{A}\subset\mathbb{C}^n\longrightarrow\mathbb{C}^m\ \ ;\ \ \ v^j \in \mathcal{C}(\mathcal{A})\ \ \forall j \in \{1,...,m\}\ ;
$$

\newtheorem{nubarv}[sets]{Definition}
\begin{nubarv}
The $\sup$-norm for $v$ is defined as follows:
$$
\nub{v}{\mathcal{A}} := \sup_{j\in\{1,...,m\}}\nub{v^j}{\mathcal{A}}:= \sup_{j\in\{1,...,m\}}\sup_{z\in\mathcal{A}}|v^j(z)|\ .
$$ 
 
\end{nubarv}
\
\newline
The shorthands
$$
\nub{\cdot}{\alpha-\beta}:=\nub{\cdot}{\dom{\alpha-\beta}}
\ \ ,\ \ \ 
\nub{\cdot}{\alpha,\beta}:=\nub{\cdot}{\dom{\alpha,\beta}} 
$$
will often be used both for functions and vector fields.
\newline
Let now $\mathcal{M}\subset\mathbb{C}^8$ be a symplectic complex manifold with local Darboux coordinates $(I_j,\vartheta_j,x_j,y_j),\ j\in\{1,2\},$ for the Liouville form
$$
\omega = \sum_{j=1}^2 dI_j \wedge d\vartheta_j + \sum_{j=1}^2 dx_j\wedge dy_j 
$$
and $F$ a hamiltonian function defined on $\mathcal{M}$. 
\newtheorem{sym_grad}[sets]{Definition}
\begin{sym_grad}
For $j\in\{1,2\}$, we will denote the {\itshape symplectic gradient} of $F$ with 
$$
X_F:= \mathcal{J}\nabla F:= \left(
\displaystyle X_F^{I_j},
\displaystyle X_F^{x_j},
\displaystyle X_F^{\vartheta_j},
\displaystyle X_F^{y_j},
\right)^\dag
:= \left(
\displaystyle -\der{F}{\vartheta_j},
\displaystyle -\der{F}{y_j},
\displaystyle\der{F}{I_j},
\displaystyle\der{F}{x_j},
\right)^\dag\ \ \ ,
$$
where 
$$
\mathcal{J}:=
\left(
\begin{matrix}
0_{4\times 4} & -\mathcal{I}_{4\times 4}\\
\mathcal{I}_{4\times 4} & 0_{4\times 4}\\
\end{matrix}
\right)
$$
is the symplectic matrix.
\end{sym_grad}
\ 
\newline
Moreover, the following anisotropic norms turn out to be particularly useful when dealing with analytic vector fields whose analyticity widths $r,s,u$ have different magnitudes.
\newtheorem{aniso}[sets]{Definition}
\begin{aniso}
To any holomorphic hamiltonian vector field $X_F$ defined in $\dom{r,s,u}\subset\mathcal{M}$ we associate the anisotropic norms 
$$
\ndb{X_F}{r,s,u}:=\max_{j\in\{1,2\}}\left\{\frac{\nub{X_F^{I_j}}{r,s,u}}{r},\ \frac{\nub{X_F^{\vartheta_j}}{r,s,u}}{s}\right\}
$$
and
$$
\ntb{X_F}{r,s,u}:=\max_{j\in\{1,2\}}\left\{\frac{\nub{X_F^{x_j}}{r,s,u}}{u},\ \frac{\nub{X_F^{y_j}}{r,s,u}}{u}\right\}\ .
$$
\end{aniso}
\ 
\newline
{\bfseries Remark: } it is easy to see that the terms in brackets have the same order of magnitude. Indeed, by making use of the Cauchy inequalities one has
\begin{align}
\begin{split}
\frac{\nub{X_F^{I_j}}{r,s,u}}{r} \leq \frac{\nub{F}{2r,2s,2u}}{rs}\ \ ; &\ \ \ \ 
\frac{\nub{X_F^{\vartheta_j}}{r,s,u}}{s} \leq \frac{\nub{F}{2r,2s,2u}}{rs}\\
\ \\
\frac{\nub{X_F^{x_j}}{r,s,u}}{u} \leq \frac{\nub{F}{2r,2s,2u}}{u^2}\ \ ; &\ \ \ \ 
\frac{\nub{X_F^{y_j}}{r,s,u}}{u} \leq  \frac{\nub{F}{2r,2s,2u}}{u^2}\ .\\
\end{split}
\end{align}
\ 
\newline
Finally, we set some notations that will be used in the next sections  when dealing with hamiltonian flows. 
\newtheorem{flows}[sets]{Definition}
\begin{flows}
The symplectic flow at time $t$ associated to the hamiltonian function $F$ is denoted with $\Lambda_F^t$ and, if such flow has period $T$, the average on $\Lambda_F^t$ of any continuous function $G$ is indicated with
$$
\langle G \rangle_F:=\frac{1}{T}\int_0^T G\circ\Lambda_F^t\ dt\ .
$$
\end{flows}
\
\newline
With the definitions above, we are now ready to build up a suitable hamiltonian framework for the planetary three-body problem.
\section{The plane, planetary three-body problem}\label{tbpp}
\subsection{Hamiltonian framework}\label{framework}
From the mathematical point of view, the planetary three-body problem consists of three points of masses $m_j,\ j\in\{0,1,2\},$ which mutually interact through the sole gravitational force. Throughout this work, the mass $m_0$ of the first body is assumed to be much greater than $m_1$ and $m_2$; for example, when considering a simplified model of the Solar System, $m_0$ represents the Sun mass whereas $m_1,m_2$ are the masses of the two major planets, i.e. Jupiter and Saturn.
\newline
By choosing the center of mass $O$ as the origin of an intertial frame, the position of the $j$-th body is given by the vector 
$$
u_j:=
\left(
u_{j1},
u_{j2},
u_{j3}
\right)^\dag\ .
$$
With this choice of coordinates, the planetary three-body hamiltonian reads
\begin{equation}
H_{init.}(\tilde{u},u):= \sum_{j=0}^2 \frac{||\tilde{u}_j||^2}{2m_j}-G\sum_{0\leq j<k\leq 2}\frac{m_j m_k}{||u_j-u_k||}\ ,
\end{equation}
where 
$$
\tilde{u}:=m_j \dot{u}_j=
\left(
\tilde{u}_{j1},
\tilde{u}_{j2},
\tilde{u}_{j3}
\right)
$$ 
are the momenta conjugated to $u_j$ for the symplectic form 
$$
\omega := \sum_{j=0}^2\sum_{k=1}^3 d\tilde{u}_{jk}\wedge du_{jk}\ .
$$
The Jacobi system of coordinates turns out to be particularly useful when studying the three-body problem. Its detailed construction may be found, for example, in the second chapter of volume I of Poincar\'{e}'s {\itshape Le\c{c}ons} \cite{Poincare_1907} or in \cite{Fejoz_2002} for a modern presentation. Here, we just give the explicit expression which links the Jacobi coordinates to the old ones
\begin{equation}\label{helio}
\left(
\begin{matrix}
r_0\\
\ \\
r_1\\
\ \\
r_2
\end{matrix}
\right):=
\left(
\begin{matrix}
1 & 0 & 0 \\
\ \\
-1 & 1 & 0\\
\ \\
-\sigma_0 & -\sigma_1 & 1\\
\end{matrix}
\right)
\left(
\begin{matrix}
u_0\\
\ \\
u_1\\
\ \\
u_2\\
\end{matrix}
\right)
=\mathcal{A}
\left(
\begin{matrix}
u_0\\
\ \\
u_1\\
\ \\
u_2\\
\end{matrix}
\right)\ ,
\end{equation}
where we have introduced the quantities
\begin{equation}
\sigma_0:=\frac{m_0}{m_0+m_1}\ ,\ \ \ \sigma_1:=\frac{m_1}{m_0+m_1}\ .
\end{equation}
The transformation can be symplectically completed for the momenta and yields
\begin{equation}
\left(
\begin{matrix}
\tilde{r}_0\\
\ \\
\tilde{r}_1\\
\ \\
\tilde{r}_2\\
\end{matrix}
\right)\ 
:=
\left(\mathcal{A}^\dag\right)^{-1}
\left(
\begin{matrix}
\tilde{u}_0\\
\ \\
\tilde{u}_1\\
\ \\
\tilde{u}_2\\
\end{matrix}
\right)\ 
=
\left(
\begin{matrix}
1 & 1 & 1 \\
\ \\
0 & 1 & \sigma_1\\
\ \\
0 & 0 & 1\\
\end{matrix}
\right)
\left(
\begin{matrix}
\tilde{u}_0\\
\ \\
\tilde{u}_1\\
\ \\
\tilde{u}_2\\
\end{matrix}
\right)\ .
\end{equation}
If we denote
\begin{align}
\begin{split}
\mu_1:=\frac{m_0m_1}{m_0+m_1}\ ,&\ \ \mu_2:=\frac{(m_0+m_1)m_2}{m_0+m_1+m_2}\\
M_1:=m_0+m_1\ ,&\ \ M_2:=m_0+m_1+m_2\ ,
\end{split}
\end{align}
then the three-body hamiltonian expressed in Jacobi coordinates assumes the following form 
\begin{align}
\begin{split}
H_{J}(r,\tilde{r})= & \sum_{j=1}^2 \frac{||\tilde{r}_j||^2}{2\mu_j}-G_N\sum_{j=1}^2 \frac{\mu_jM_j}{||r_j||}\\
+ & G_N m_2\left(\frac{M_1}{||r_2||}-\frac{m_0}{||r_2+\sigma_1r_1||}-\frac{m_1}{||r_2-\sigma_0r_1||}\right)\ .
\end{split}
\end{align}
The first part of the hamiltonian describes the keplerian motion of two bodies of masses $\mu_j$ around a central attractor of mass $M_j$, whereas the second row has a much smaller magnitude and can be treated as a perturbation.
Indeed, by defining 
\begin{equation}
K(\tilde{r}_j,r_j):= \sum_{j=1}^2 \frac{||\tilde{r}_j||^2}{2\mu_j}-G_N\sum_{j=1}^2 \frac{\mu_jM_j}{||r_j||}\ ,
\end{equation}
\begin{equation}\label{k2}
P(r_j):= G_N m_2\left(\frac{M_1}{||r_2||}-\frac{m_0}{||r_2+\sigma_1r_1||}-\frac{m_1}{||r_2-\sigma_0r_1||}\right)
\end{equation}
and 
$$
\varepsilon:=	\max_{j\in \{1,2\}}\left\{\varepsilon_j\right\} := \max_{j\in \{1,2\}}\left\{\frac{m_j}{m_0}\right\}\ ,
$$
it is straightforward to see that
$$
\left|\frac{P(r_j)}{K(\tilde{r}_j,r_j)}\right|_{r,s,u}=O(\varepsilon)\ .
$$
In the case of the Sun-Jupiter-Saturn system one has $\varepsilon\sim 10^{-3}$.
\newline
As it is well known (see e.g. \cite{Boccaletti_Pucacco_2004} for a detailed explanation), the unperturbed keplerian problem described by hamiltonian $K_1$ satisfies the hypotheses of the Arnold-Liouville integrability theorem. Namely, for negative values of the total energy, its trajectories in the configuration space are two fixed ellipses (labeled with an index $j\in\{1,2\}$). The semimajor axes and eccentricities are denoted, respectively, with $a_j$ and $e_j$ and the position of the orbit with respect to a plane of reference is described by the three Euler angles which, in this particular case, are the longitude of the ascending node $\Omega_j$, the argument of periapsis $\omega_j$ and the inclination $\iota_j$. The position of a body along its elliptic trajectory is determined once its real anomaly $f_j$ is given. 
\begin{figure}[h]\label{elements}
\centering
\hspace*{-0.cm}\includegraphics[scale=0.25]{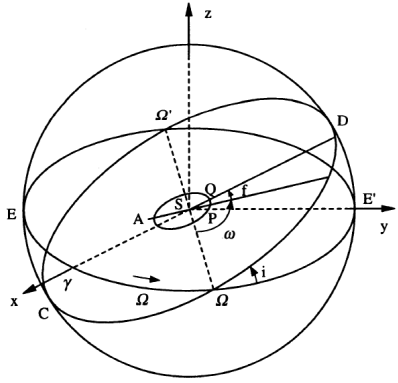}
\caption{Orbital elements for the Kepler's problem (from Boccaletti \& Pucacco \citep{Boccaletti_Pucacco_2004}).}
\end{figure}
\begin{figure}[h]\label{eccentric}
\centering
\hspace*{-0.cm}\includegraphics[scale=0.3]{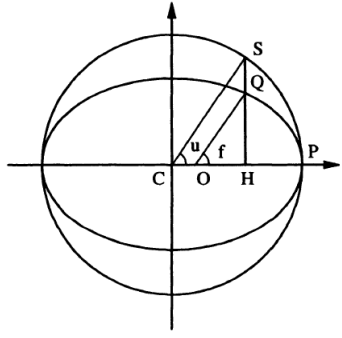}
\caption{Diagram showing the eccentric anomaly $u$ and the real anomaly $f$; the celestial body is at point Q and the Sun is in O (from Boccaletti \& Pucacco \citep{Boccaletti_Pucacco_2004}).}
\end{figure}
\
\newline
We denote with $n_j$ the mean motion (frequency of the real anomaly) of the $j$-th body and we define the mean anomalies
$$
M_j:= n_j (t-t_0)\ ,
$$
which are related to the eccentric anomalies $u_j$ by Kepler's equation
$$
M_j=u_j-e_j\sin{u_j}\ .
$$
As a consequence of Arnold-Liouville integrability theorem, action-angle coordinates exist for this problem and are named after French's mathematician and astronomer Charles Delaunay, who first introduced them in his {\itshape Trait\'e du mouvement de la lune}, published in 1860 (see \cite{Delaunay_1860}).
The Delaunay variables written as functions of the orbital elements read
\begin{equation}\label{delaunay}
\begin{cases}
\Lambda_j &:= \mu_j\sqrt{G_NM_j a_j}\\
G_j &:= \Lambda_j \sqrt{1-e_j^2}\\
\Theta_j &:= G_j \cos \iota_j\\
l_j &:= M_j\\
g_j &:=\omega_j\\
\theta_j &:= \Omega_j
\end{cases}
\end{equation}
and it is straightforward to see that they are ill-defined for null eccentricities and inclinations. Therefore, another system of coordinates must be considered in order to avoid singularities. The usual choice consists in introducing Poincar\'{e}'s elliptic variables, namely
\begin{equation}\label{poincare}
\begin{cases}
\Lambda_j & :=  \mu_j\sqrt{G_N M_j a_j}\\
\lambda_j & :=  M_j + \omega_j+\Omega_j\\
x_j+iy_j & := \left[2\Lambda_j\left(1- \sqrt{1-e_j^2}\right)\right]^{1/2}\exp[-i(\omega_j+\Omega_j)]\\
p_j+iq_j & := \left[2\Lambda_j\sqrt{1-e_j^2}(1-\cos{\iota_j})\right]^{1/2}\exp(-i\Omega_j)
\end{cases}\ .
\end{equation}
In this frame, the planetary three-body hamiltonian takes the form (the superscript $p$ stands for Poincar\'{e})
\begin{equation}\label{H_poincare}
H^p(\Lambda_j,\lambda_j,x_j,y_j,p_j,q_j)=H^p_K(\Lambda_j)+\varepsilon H^p_P(\Lambda_j,\lambda_j,x_j,y_j,p_j,q_j)\ ,\ \ j\in\{1,2\}\ ,
\end{equation}
where the Keplerian part just depends on the actions $\Lambda_j$
\begin{equation}\label{kepler}
H^p_K(\Lambda):= -G_N^2\sum_{j=1}^2 \frac{M_j^2\mu_j^3}{2\Lambda_j^2}
\end{equation} 
and the perturbation can be explicitly computed by inserting system (\ref{poincare}) into (\ref{k2}).
For more details about the Poincar\'{e} variables, see e.g. \cite{Fejoz_2013}.
\newline
For $\varepsilon=0$, the phase space of the unperturbed system is foliated with invariant tori. 
Now, choose two fixed actions $\Lambda_1^0$ and $\Lambda_2^0$ corresponding to a resonant frequency vector $\omega:=(\omega_1,\omega_2)$ for the unperturbed system, i.e. 
$$
\frac{\omega_1}{\omega_2}=\frac{p}{q}\ ,
$$
with $p$ and $q$ two positive integers. We are interested in the behaviour of the planetary three-body hamiltonian in the neighborhood of the resonant torus corresponding to these frequencies, so we consider the translation
\begin{equation}
\left(
\begin{matrix}
I_1\\
I_2
\end{matrix}
\right)
:=
\left(
\begin{matrix}
\Lambda_1-\Lambda_1^0\\
\Lambda_2-\Lambda_2^0
\end{matrix}
\right)
\end{equation} 
and we compute a Taylor's developement of $H^p_K$ with initial point $(I_1,I_2)=(0,0)$. \newline
As a matter of notation, in the sequel we shall often use the shorthand $(I,\vartheta,x,y,\xi,\eta)$ to denote $(I_1,I_2,\vartheta_1,\vartheta_2,x_1,x_2,y_1,y_2,\xi_1,\xi_2,\eta_1,\eta_2)$.
\newline
Now, we restrict to the planar case $(\xi,\eta)=(0,0)$, so that the complete hamiltonian assumes the form
\begin{align}\label{ham_bello}
\begin{split}
H(I,\vartheta,x,y)= & H_K(I)+\varepsilon H_P(I,\vartheta,x,y)\\
= & H_K(0)+\langle \omega , I \rangle + \mathcal{G}(I) +\varepsilon H_P(I,\vartheta,x,y)\ ,
\end{split}
\end{align}
where in the second line we have performed a Taylor expansion and $\mathcal{G}(I)$ denotes the remainder of order $2$ in the actions.
\newline
If we denote 
$$
h(I):=\langle \omega , I \rangle\ ,
$$
the hamiltonian can be splitted into a resonant part $g_0$ and a non-resonant part $f_0$
\begin{equation}
g_0(I,\vartheta):=\mathcal{G}(I)+\varepsilon \left\langle H_P(I,\vartheta,x,y) \right\rangle_h\ ,
\end{equation}
\begin{equation}
f_0(I,\vartheta):=\varepsilon H_P(I,\vartheta,x,y)-\varepsilon\left\langle H_P(I,\vartheta,x,y)\ \right\rangle_h\ 
\end{equation}
which, from their very definitions, satisfy $\langle f_0 \rangle_h=0$ and $\{h,g_0\}=0$.
\newline
As we shall see in the next paragraph, our purpose consists in reducing the size of $f_0$ with the help of some sharp techniques of perturbation theory.

\subsection{Analyticity widths, convexity and initial estimates}\label{analiticity}
It is well known that hamiltonian (\ref{H_poincare}) is analytic in some complex domain and, as we shall see later on, a good knowledge on the analyticity widths is crucial in establishing the limits of the theory we deal with. Here, we rely on the recent and important work \cite{Castan_2017} by Castan which gives explicit estimates for the magnitude of hamiltonian (\ref{H_poincare}) in its domain of analyticity. We stress the fact that in \cite{Castan_2017} the analyticity of the {\itshape complete} hamiltonian is taken into account, without making any truncation, so that one is left with estimates on the analyticity widths which take into account all the singularities that function (\ref{H_poincare}) encounters in the complex field. Explicit numerical values will be considered in paragraph \ref{num_comp}; here, we shall just assume that hamiltonian (\ref{ham_bello}) is analytic in a domain $\mathcal{D}_{\rho,4r,4s,\xi,4u}$ for some $(\rho,r,s,\xi,u)\in\mathbb{R}^5$. Furthermore, since the unperturbed hamiltonian (\ref{kepler}) is continuous and convex on the bounded domain we are considering, for all couples $(I_1,I_2)\in S_I(4r)\times S_I(4r)$  the eigenvalues $\varrho_1(I),\varrho_2(I)$ of the hessian matrix $D^2 H_K(I)$ satisfy
\begin{align*}
\begin{split}
&|\varrho_1|_{r,s,u}+|\varrho_2|_{r,s,u}\leq K\\
&\min\{|\varrho_1|_{r,s,u},|\varrho_2|_{r,s,u}\}\geq \kappa
\end{split}\ ,
\end{align*}
where $\kappa,K$ are two positive real constants which can be computed explicitly since the expression for $H_K$ is explicit.
As we shall see in paragraph \ref{stab}, convexity plays a crucial role in insuring stability.
\newline
Finally, we estimate the sizes of functions and vector fields by making use of the Cauchy inequalities:
\begin{equation}\label{useful}
\begin{cases}
\nub{f_0}{4} := \nub{\varepsilon H_P-\varepsilon\langle H_P\rangle_h}{4}  \leq 2\varepsilon\nub{H_P}{4}\ \\
\ \\
\nub{g_0-\mathcal{G}}{4} := \nub{\varepsilon\langle H_P\rangle_h}{4} \leq \varepsilon\nub{H_P}{4}\ \\
\ \\
\ndb{X_{f_0}}{3} := \displaystyle\max_{j\in\{1,2\}}\left\{\frac{\nub{\ch{f_0}^{I_j}}{3}}{r},\frac{\nub{\ch{f_0}^{\vartheta_j}}{3}}{s}\right\}\leq \frac{\nub{f_0}{4}}{rs}=: \eo\\
\ \\
\ntb{X_{f_0}}{3} := \displaystyle \max_{j\in\{1,2\}}\left\{\frac{\nub{\ch{f_0}^{x_j}}{3}}{u},\frac{\nub{\ch{f_0}^{y_j}}{3}}{u}\right\}\leq \frac{\nub{f_0}{4}}{u^2} =: \Eo\\
\ \\
\ndb{X_{g_0}-X_{\mathcal{G}}}{3} := \displaystyle \max_{j\in\{1,2\}}\left\{\frac{\nub{\ch{g_0-\mathcal{G}}^{I_j}}{3}}{r},\frac{\nub{\ch{g_0-\mathcal{G}}^{\vartheta_j}}{3}}{s}\right\}\leq \frac{\nub{g_0-\mathcal{G}}{4}}{rs}=: \gamma_0\\
\ \\
\ntb{X_{g_0}-X_{\mathcal{G}}}{3} := \displaystyle \max_{j\in\{1,2\}}\left\{\frac{\nub{\ch{g_0-\mathcal{G}}^{x_j}}{3}}{u},\frac{\nub{\ch{g_0-\mathcal{G}}^{y_j}}{3}}{u}\right\}\leq \frac{\nub{g_0-\mathcal{G}}{4}}{u^2} \leq =: \Gamma_0
\ \\
\ndb{X_{\mathcal{G}}}{3} := \displaystyle\frac{\nub{\ch{\mathcal{G}}^{\vartheta_j}}{3}}{s}:= \dd\\
\end{cases}\ .
\end{equation}
\ 
\newline
Notice that since $\ch{\mathcal{G}}$ has an explicit expression in the case we are considering, it is directly estimated without making use of the Cauchy inequalities.
\newline
As we see from the estimates above, $u=\sqrt{rs}$ is a natural choice for the analyticity width in the cartesian variables. However, since we want to stay as sharp as possible, we choose to leave $u$ as a free parameter and we set
\begin{equation}\label{beta_def}
\beta := \frac{\sqrt{rs}}{u}\ .
\end{equation}
\
\newline
With this setup, we can now define three real functions $\upsilon_0,\Upsilon_0,\zeta_0: \mathbb{R}\longrightarrow\mathbb{R}$ which depend on the parameters $\eo,\Eo,\go,\Go,\dd$ and act as follows:
\begin{align}\label{def_upsilon}
\begin{split}
\upsilon_0:\ x  \longmapsto  & \left(Tx\right)^2\eo\chi_0+Tx^2\Theta_0(2\eo+2\go+\dd)+\frac{Tx}{2}\chi_0\\ 
& + \Theta_0 x\left(1+\frac{\go}{\eo}+\frac{\dd}{\eo}\right)+\left(\frac{Tx\Eo}{\beta}\right)^2\frac{\chi_0}{\eo}\\
& +2\frac{Tx^2\Eo\Theta_0}{\beta^2}\left(\frac{\Eo}{\eo}+\frac{\Go}{\eo}\right)\ ,
\end{split}
\end{align}
\begin{align}\label{def_Upsilon}
\begin{split}
\Upsilon_0: \ x  \longmapsto & \beta^2\left(Tx\eo\right)^2\frac{\chi_0}{\Eo}+\beta^2Tx^2\eo\Theta_0\left(2\frac{\eo}{\Eo}+2\frac{\go}{\Eo}+\frac{\dd}{\Eo}\right)+\frac{Tx}{2}\chi_0\\
& + \Theta_0 x\left(1+\frac{\Go}{\Eo}\right)+\left(Tx\Eo\right)^2\frac{\chi_0}{\Eo}+2Tx^2\Theta_0(\Eo+\Go)\ ,\\
\ \\
\end{split}
\end{align}
\begin{align}\label{def_zeta0}
\zeta_0:\ x\longmapsto & \frac{Tx}{2}\max\{\go+\dd,\Go\}+\Theta_0 x\ ,
\end{align}
where $\chi_0$ and $\Theta_0$ are two real constants which read
\begin{align}\label{chi0}
\begin{split}
\chi_0
:= & \max\{\Eo+\Go,\eo+\go+\dd\}\ ,\\
\ \\
\Theta_0:= &\max\left\{\frac{T\Eo}{2},\frac{T\eo}{2}\right\}\ .
\end{split}
\end{align}
In the sequel, $\upsilon_0,\Upsilon_0$ will describe the decreasing of the vector field associated to the non resonant perturbation, while $\zeta_0$ is related to the decreasing of the non-resonant perturbation itself. \newline
With the construction above, we can exploit the convexity of the integrable part of the hamiltonian in order to obtain a theorem that insures stability in the action variables for a suitably long time. To do this, we shall construct a sharp resonant normal form inspired by a result contained in \cite{Poschel_1999} and then we shall confine the actions with the help of a geometric tool described in \cite{Lochak_1992} and \cite{Lochak_1993}. We stress that the estimates and the techniques which will henceforth be used can be generalized to any quasi-integrable system. Furthermore, in the case under study, the drift of the cartesian variables $(x_j,y_j)$ will be bounded by the conservation of the total angular momentum
\begin{equation}\label{angular_momentum}
\mathcal{N}:=\sum_{j=1}^2\Lambda_j(t)\sqrt{1-e_j^2(t)}\ .
\end{equation}
\subsection{Stability in the neighbourhood of a periodic torus}\label{stab}
The main theorem can be stated as follows:
\newtheorem{ac_stability}{Theorem (Stability for the whole system)}
\begin{ac_stability}\label{ac_stability}
Assume the previous constructions and definitions for hamiltonian (\ref{ham_bello}).\\ Suppose that there exist $m\in\mathbb{N}$ and three numbers $p,q_1,q_2\in\left]0,\displaystyle\frac{2}{3}\right[$
satisfying 
\begin{align}
2\upsilon_0(m)<q_1\ ,\ \ \ 
2\Upsilon_0(m)<q_2\ ,\ \ \ 
2\zeta_0(m)< p\ .
\end{align}
Fix $\varepsilon$ sufficiently small so that one can pick two positive real numbers $R$, $\xi_0$ such that
\begin{align}\label{size_initial_rad}
\begin{split}
C_1(R) & >0\\
\xi+\left(1-\frac{T\Eo}{2}\frac{1-q_2^m}{1-q_2}\right)u & >\xi_0\geq 0\\
\xi+\left(1-\frac{T\Eo}{2}\frac{1-q_2^m}{1-q_2}\right)u & \geq \sqrt{\Lambda_1^0+\Lambda_2^0+2\left(\rho+r+\frac{T\eo}{2}\frac{1-q_1^m}{1-q_1}r\right)-\mathcal{N}^-(\xi_0)}
\end{split},
\end{align}
where $C_1(R)$ denotes the quantity
\begin{align}\label{cunooo}
\begin{split}
 & \frac{\kappa}{2}  \left\{\left[\rho+r-\left(\frac{K}{\kappa}+1\right)\left(R+\frac{T\eo}{2}\frac{1-q_1^m}{1-q_1}r\right)\right]^2 -\left[\frac{K}{\kappa}\left(R+\frac{T\eo}{2}\frac{1-q_1^m}{1-q_1}r\right)\right]^2 \right\}\\
    & -\left(p\frac{1-p^m}{1-p}+2p^m\right) \nub{f_0}{3}-2\nub{g_0-\mathcal{G}}{3}
    \end{split}
\end{align}
and we have defined
\begin{align}
\begin{split}
\mathcal{N}^-(\xi_0) & :=(\Lambda_1^0-R)\sqrt{1-\bar{e}_{1}(0,\xi_0)^2}+(\Lambda_2^0-R)\sqrt{1-\bar{e}_{2}(0,\xi_0)^2}\ ,\\
\bar{e}_j(0,\xi_0) & :=\sqrt{1-\left(1-\frac{\xi_0^2}{2(\Lambda^0_1-R)}\right)^2}\ ,\ \ j\in\{1,2\}.
\end{split}
\end{align}
\ 
\newline
Then, for any initial condition 
\begin{equation}
(I(0),\vartheta(0),x(0),y(0))\in S_0\left(R\right)\times S_0\left(R\right)\times \mathbb{T}^2\times B_{0,0}(\xi_0)\times B_{0,0}(\xi_0)
\end{equation}
the flow of hamiltonian (\ref{ham_bello}) stays in the domain of analyticity $\mathcal{D}_{\rho,r,s,\xi,u}$ and there exist a positive constant $C_2$ and three functions $R_f:\mathbb{R}\longrightarrow\mathbb{R}$, $(\overline{e}_1(t,\xi_0),\overline{e}_2(t,\xi_0)):\mathbb{R}\longmapsto ]0,1[\times]0,1[$ such that for any time
\begin{equation}
|t| < \bar{t}:= \frac{C_1(R)}{C_2}q_1^{-m}\ ,
\end{equation}
one has
\begin{align}
\begin{split}
\nub{I(t)-I(0)}{S_0(R)\times S_0\left(R\right)} \leq R_f(t)&\\
\ \\
e_1(t)<\overline{e}_1(t,\xi_0)&\\ 
\ \\
e_2(t)<\overline{e}_2(t,\xi_0)&\ .
\end{split}
\end{align}
Moreover, such constant and functions can be can be computed explicitly and read:
\begin{align}
\begin{split}
C_2 := & r|\omega_1+\omega_2|\eo\ ,\\
\ \\
R_f :\ &  \ t\longmapsto  \frac{K}{\kappa} \tr +\sqrt{\left(\frac{K}{\kappa} \tr\right)^2+a(t)}+\frac{T\eo}{2}\frac{1-q_1^m}{1-q_1}r\ ,\\
\ \\
\overline{e}_1:\ & (t,\xi_0)\longmapsto \sqrt{1-\left(\frac{\mathcal{N}^{-}(\xi_0)-\Lambda_2^0-R_f(t)}{\Lambda_1^0+R_f(t)}\right)^2}\ ,\\
\ \\
\overline{e}_2:\ &  (t,\xi_0)\longmapsto \sqrt{1-\left(\frac{\mathcal{N}^{-}(\xi_0)-\Lambda_1^0-R_f(t)}{\Lambda_2^0+R_f(t)}\right)^2}\ .
\end{split}
\end{align}
where we have defined  
\begin{align*}
a(t):=&\frac{2}{\kappa}\left[\left(p\frac{1-p^m}{1-p}+2p^m\right) \nub{f_0}{3}+2\nub{g_0-\mathcal{G}}{3}\ + C_2\ q_1^m|t|\right]\\
\tr:=&R+\frac{T\eo}{2}\frac{1-q_1^m}{1-q_1}r\ .
\end{align*}

\end{ac_stability}
\
\\
The proof of such result can be split into two parts which insure, respectively, stability in the action variables and confinement in the cartesian ones.
\subsubsection{Confinement of the actions}
\newtheorem{stability}[ac_stability]{Theorem (Stability of the action variables)}
\begin{stability}\label{stability}
Assume the constructions of the previous section for hamiltonian (\ref{ham_bello}).\\
Suppose that there exist $m\in\mathbb{N}$ and three numbers $p,q_1,q_2\in\left]0,\displaystyle\frac{2}{3}\right[$
satisfying 
\begin{align}\label{hyp_stab_2}
2\upsilon_0(m)<q_1\ ,\ \ \ 
2\Upsilon_0(m)<q_2\ ,\ \ \ 
2\zeta_0(m)< p\ .
\end{align}
Fix $\varepsilon$ sufficiently small so that one can pick a positive real number $R$ such that
\begin{align}\label{size_initial_rad}
\begin{split}
C_1(R) & >0\\
\end{split},
\end{align}
where $C_1(R)$ is defined as in (\ref{cunooo}).
\ 
\newline
Then there exist a positive constant $C_2$ and a function $R_f:\mathbb{R}\longrightarrow\mathbb{R}$ such that, for a solution with initial actions satisfying
\begin{equation}
I(0)\in S_0\left(R\right)\times S_0\left(R\right)
\end{equation}
and for any time 
\begin{equation}
|t| < \min\left\{\frac{C_1(R)}{C_2}q_1^{-m},t_{esc}\right\}\ ,
\end{equation}
where $t_{esc}$ is the time of escape from $\mathcal{D}_{1-\frac{T\eo}{2}\frac{1-q_1^m}{1-q_1},1-\frac{T\Eo}{2}\frac{1-q_2^m}{1-q_2}}$,
one has
\begin{align}
\begin{split}
\nub{I(t)-I(0)}{S_0(R)\times S_0\left(R\right)} \leq R_f(t)\ . \\ 
\end{split}
\end{align}
Moreover, the constant $C_2$ and the function $R_f$ can be explicitly computed and read:
\begin{align}\label{C1}
\begin{split}
C_2 := & r|\omega_1+\omega_2|\eo\ ,\\
\ \\
R_f :\ \ &  \ t\longmapsto  \frac{K}{\kappa} \tr +\sqrt{\left(\frac{K}{\kappa} \tr\right)^2+a(t)}+\frac{T\eo}{2}\frac{1-q_1^m}{1-q_1}r\ ,\\
\end{split}
\end{align}
where we have denoted  
\begin{align}
\begin{split}
a(t) & :=\frac{2}{\kappa}\left[\left(p\frac{1-p^m}{1-p}+2p^m\right) \nub{f_0}{3}+2\nub{g_0-\mathcal{G}}{3}\ + C_2\ q_1^m|t|\right]\\
\tr & := R+\frac{T\eo}{2}\frac{1-q_1^m}{1-q_1}r\ .
\end{split}
\end{align}

\end{stability}
\ 
\newline
In order to prove theorem \ref{stability}, one must firstly put hamiltonian (\ref{ham_bello}) into resonant normal form by applying a transformation which is described in the next
\newtheorem{normal_form}[ac_stability]{Lemma (Resonant Normal Form)}
\begin{normal_form}\label{nf_lemma}
Let $H_0$ be an hamiltonian function, analytical in $\mathcal{D}_{3}$, which can be decomposed as follows:
$$
H_0=h+g_0+f_0\ ,
$$
where $h:= \langle \omega,I \rangle$ is integrable with a $T-$periodic frequency vector $\omega$, $g_0$ is in involution with $h$, i.e. $\{h,g_0\}=0$, and $\langle f_0 \rangle_h = 0$.
\newline
Assume the existence of an integrable hamiltonian $\mathcal{G}$ and of five real numbers, $\eo$, $\go$, $\dd$, $\Eo$ and $\Go$, such that
\begin{align}\label{hyp_nf_0}
\begin{split}
\ndb{\ch{f_0}}{3} \leq \eo,\ \ \ \ndb{\ch{g_0}-\ch{\mathcal{G}}}{3} \leq \go,\ \ \ \ndb{\ch{\mathcal{G}}}{3}\leq \dd\\
\ntb{\ch{f_0}}{3} \leq \Eo,\ \ \ \ndb{\ch{g_0}-\ch{\mathcal{G}}}{3} \leq \Go\ .
\end{split}
\end{align}
Assume, also, that there exist $m\in\mathbb{N}$ and three numbers $p,q_1,q_2\in\left]0,\displaystyle\frac{2}{3}\right[$
satisfying 
\begin{equation}\label{hyp_nf_2}
2\upsilon_0(m)<q_1\ ,\ \ 2\Upsilon_0(m)<q_2\ ,\ \ 2\zeta_0(m)< p\ .
\end{equation}
Then there exist a symplectic transformation $\Psi_m$, analytic and real-valued for any real argument
$$
\Psi_m: \dom{1}\longrightarrow\dom{1+\frac{T\eo}{2}\frac{1-q_1^m}{1-q_1},1+\frac{T\Eo}{2}\frac{1-q_2^m}{1-q_2}}\ ,
$$
whose size is 
\begin{equation}\label{size_nf}
\ndb{\Psi_m-id}{1}\leq \frac{T\eo}{2}\frac{1-q_1^m}{1-q_1}\ ,\ \ \ \ntb{\Psi_m-id}{1}\leq \frac{T\Eo}{2}\frac{1-q_2^m}{1-q_2}\ ,
\end{equation}
such that
$$
H_m:= H_0\circ\Psi_m=h+g_m+f_m\ ,
$$
where $\{h,g_m\}=0\ $ and $\langle f_m\rangle_h=0\ $.
\newline
Furthermore, the following estimates hold
\begin{align}\label{ch_nf}
\begin{split}
\ndb{\ch{f_m}}{1}\leq \ q_1^m\eo \ \ \ &\ndb{\ch{g_m}-\mathcal{G}}{1}\leq \ \go+\frac{q_1}{2}\frac{1-q_1^m}{1-q_1}\eo\\
\ntb{\ch{f_m}}{1}\leq \ q_2^m\Eo \ \ \ &\ntb{\ch{g_m}-\mathcal{G}}{1}\leq \Go+\frac{q_2}{2}\frac{1-q_2^m}{1-q_2}\Eo\ ,\\
\nub{f_m}{1}\leq p^m\nub{f_0}{3} \ \ \ &\nub{g_m-\mathcal{G}}{1}\leq \ \frac{p}{2}\frac{1-p^{m}}{1-p}\nub{f_0}{3}+\nub{g_0-\mathcal{G}}{3}\ .\\
\end{split}
\end{align}
\end{normal_form}
\
\newline
This lemma is proven by iterating $m$ times the following result which is, in turn, an improved version of a result contained in \cite{Poschel_1999}. All constant are made explicit here and we have tried to sharpen all the estimates as much as possible.
\newtheorem{iterative}[ac_stability]{Lemma (Single perturbative iteration)}
\begin{iterative}\label{iterative}
Let $H_0$ be a hamiltonian function, analytical in $\mathcal{D}_{1}$, for which the following decomposition holds:
$$
H_0=h+g_0+f_0\ ,
$$
where $h:= \langle \omega,I \rangle$ is integrable and has a $T-$periodic frequency vector $\omega$, $g_0$ is in involution with $h$, i.e. $\{h,g_0\}=0$, and $\langle f_0 \rangle_h = 0$.
\newline
Assume the existence of five real numbers $\eo$, $\go$, $\dd$, $\Eo$ and $\Go$ such that
\begin{align}\label{hyp_it_0}
\begin{split}
\ndbfo \leq \eo,\ \ \ \ndbgo \leq \go,\ \ \ \ndbgc\leq \dd\\
\ntbfo \leq \Eo,\ \ \ \ntbgo \leq \Go\ .
\end{split}
\end{align}
Furthermore, suppose that for a real number $\auno \in \left]0,\displaystyle\frac{1}{2}\right[$ one has 
\begin{equation}\label{hyp_it_1}
\frac{T\eo}{2\auno}< 1\ . 
\end{equation}
Then there exist a symplectic analytical transformation $\Phi_1$ with generating function $\phi_1$,
$$
\Phi_1: \dom{1-2\auno}\longrightarrow \dom{1-\auno}\ ,
$$
which is real valued for any real argument and whose size is
\begin{equation}\label{size_it}
||\Phi_1-id||_{1-2\auno}\leq \frac{T\eo}{2},\ \ \ |||\Phi_1-id|||_{1-2\auno}\leq \frac{T\Eo}{2}\ ,
\end{equation}
which takes the hamiltonian into the following form:
$$
H_1:= H_0\circ\Phi_1= h+g_1+f_1\ ,
$$
where $\{h,g_1\}=0$ and $\langle f_1\rangle_h =0$. 
\newline
\newline
As for vector fields estimates one has
\begin{equation}\label{bound_chd_f}
\ndbfuno \leq 2\upsilon_0\left(\frac{1}{\auno}\right)\ \eo,\ \ \ \ndbguno \leq \upsilon_0\left(\frac{1}{\auno}\right)\ \eo + \go
\end{equation}
and
\begin{equation}\label{bound_cht_f}
\ntbfuno \leq 2\Upsilon_0\left(\frac{1}{\auno}\right)\ \Eo,\ \ \ \ntbguno \leq \Upsilon_0\left(\frac{1}{\auno}\right)\ \Eo + \Go\ ,
\end{equation}
whereas functions are bounded by
\begin{equation}\label{bound_it_function}
\nubfuno \leq 2 \zeta_0\left(\frac{1}{\auno}\right)\ |f_0|_1,\ \ \ \nubguno \leq \zeta_0\left(\frac{1}{\auno}\right)\ |f_0|_1+\nub{g_0-\mathcal{G}}{1}\ .
\end{equation}
\end{iterative}
\ 
\newline
This lemma is proven by making use of some sharp techniques of perturbation theory.
\begin{proof}[Proof]
We look for a transformation $\Phi_1$ which is the symplectic flow at time $t=1$ of a generating function $\phi_1$, so that the original hamiltonian takes the form
\begin{align}\label{ham_trans1}
\begin{split}
H_1 & = H_0\circ\Phi_1=
e^{L_{\phi_1}}(h+g_0+f_0) \\ 
& =h+g_0+f_0
+L_{\phi_1}(h)
+\sum_{n\geq 2}\frac{1}{n!}L^n_{\phi_1}(h)+\sum_{n\geq 1}\frac{1}{n!}L^n_{\phi_1}(g_0+f_0)\ ,
\end{split}
\end{align}
and we impose the homological equation
\begin{equation}\label{homological}
\{\phi_1,h\}=-f_0\ ,
\end{equation}
whose solution is 
\begin{equation}\label{phi_1}
\phi_1=\frac{1}{T}\integ{0}{T}{tf_0\circ\Lambda_h^t}{t}=\frac{1}{T}\integ{0}{T}{tf_0(I,\vartheta+\omega t, x, y)}{t}\ .
\end{equation}
In this way, the transformed hamiltonian reads
\begin{align}\label{ham_trans2}
\begin{split}
H_1
&=h+g_0
+\sum_{n\geq 2}\frac{1}{n!}L^n_{\phi_1}(h)+\sum_{n\geq 1}\frac{1}{n!}L^n_{\phi_1}(g_0+f_0)\\
&=h+g_0+r_1\ , 
\end{split}
\end{align}
where
\begin{equation}\label{reste}
r_1:=\int_0^1\{\phi_1,g_0+tf_0\}\circ\Lambda^t_{\phi_1}\ dt
\end{equation}
is the integral form for the remainder.
\newline
Furthermore, if we define
\begin{equation}\label{def_g1}
g_1:= g_0+\langle r_1\rangle_h
\end{equation}
and 
\begin{equation}\label{def_f1}
f_1:= r_1 - \langle r_1\rangle_h\ 
\end{equation}
the following resonant decomposition holds
$$
H_1=h+g_1+f_1\ ,\ \ \{h,g_1\}=0\ ,\ \ \langle f_1\rangle_h=0\ .
$$
Now, in order to prove that the flow $\flphit$ starting from $\dom{1-2\auno}$ stays in $\dom{1-\auno}$ for $|t|\leq 1$, we define the time of escape
\begin{equation}\label{escape_time}
t^*:= \inf \{t\in \mathbb{R} \text{ s.t. } \flphit(\dom{1-2\auno})\notin\dom{1-\auno} \}
\end{equation}
and we find the following estimates for the hamiltonian vector field $X_{\phi_1}$ associated to $\phi_1$:
\begin{equation}\label{chd_phi1}
||\ch{\phi_1}||_{1}= \ndb{\frac{1}{T}\integ{0}{T}{t\ch{f_0}(I,\vartheta+\omega t,x,y)}{t}}{1}
\leq \frac{T}{2}\ndb{\ch{f_0}}{1}\leq \frac{T\eo}{2}
\end{equation}
and, analogously,
\begin{equation}\label{cht_phi1}
\ntb{\ch{\phi_1}}{1}\leq \frac{T\Eo}{2}\ .
\end{equation}
Let us consider, as an example, an escape from $\dom{1-\auno}$ of the action component of the flow with initial conditions in $\dom{1-2\auno}$, i.e.
$$
\auno r \leq \normsup{(\flphits-id)^{I_j}}{1-2\auno}\ ,\ \ j\in\{1,2\}\ .
$$
By making use of some straightforward inequalities, one easily sees that
\begin{align*}
\begin{split}
\auno r & \leq \normsup{(\flphits-id)^{I_j}}{1-2\auno} \leq \integ{0}{t^*}{|X_{\phi_1}^{I_j}|_{1-2\auno}}{t}\ \ \ \ \ j\in\{1,2\}\\
& \leq \ndb{\ch{\phi_1}}{1}r|t^*|\leq \frac{T\eo}{2}r|t^*|\ ,
\end{split}
\end{align*}
which is equivalent to
$$
\frac{T\eo}{2\auno}|t^*|\geq 1\ ,
$$
so that, by hypothesis (\ref{hyp_it_1}), one gets $|t^*|>1$.\newline
In a completely analogous way one proves a similar result for the angles $\vartheta_j$ and for the cartesian variables $(x_j,y_j)$ and is thus insured that
$$
\Phi_1:\dom{1-2\auno}\longrightarrow\dom{1-\auno}\ .
$$
The discussion above, together with estimates (\ref{chd_phi1}) and (\ref{cht_phi1}) implies
\begin{align}\label{sizephi}
\begin{split}
\ndb{\Phi_1-id}{1-2\auno} & \leq \frac{T\eo}{2}\ ,\\
\ntb{\Phi_1-id}{1-2\auno} & \leq \frac{T\Eo}{2}\ .
\end{split}
\end{align}
Finally, in order to prove estimates (\ref{bound_chd_f}) and (\ref{bound_cht_f}) in the statement, we consider the symplectic field associated to the remainder in expression (\ref{reste}), namely
\begin{align}\label{reste_ch}
\begin{split}
\ch{r_1} & = \integ{0}{1}{\jj \left(D\flphit\right)^{\dag}\left[\nabla\left(\{\phi_1,g_0+tf_0\}\right)\circ\flphit\right]}{t}\\
& =\integ{0}{1}{\jj \left(D\flphit\right)^{\dag}\jj^{-1}\jj\left[\nabla\left(\{\phi_1,g_0+tf_0\}\right)\circ\flphit\right]}{t}\\
& =\integ{0}{1}{\mathcal{M}\left(\left[\ch{\phi_1},\ch{g_0+tf_0}\right]\circ\flphit\right)}{t}\ ,
\end{split}
\end{align}
where we have defined the matrix $\mathcal{M}:=\jj \left(D\flphit\right)^{\dag}\jj^{-1}$ and we have used the fact that the symplectic gradient of a Poisson bracket yields the Lie bracket (see e.g. \cite{McDuff_Salamon_2017} for a proof of this statement). 
\newline
We show in appendix \ref{proof_est} that estimates (\ref{bound_chd_f}) and (\ref{bound_cht_f}) follow immediately from definitions (\ref{def_g1}) and (\ref{def_f1}) and from expression (\ref{reste_ch}), provided that one gives a bound to the matrix $\mathcal{M}$ with the help of the Cauchy inequalities, and a bound to the Lie bracket by making use of an argument in \cite{Fasso_1990}. A similar procedure will yield a bound on the remainder (\ref{reste}) which, in turn, will imply inequalities (\ref{bound_it_function}), as we show in appendix \ref{proof_est} as well.
\end{proof}
\
\newline
We are now able to write the proof of the normal form lemma.
\begin{proof}
This lemma is proven by iterating $m$ times the machinery described in the proof of lemma (\ref{iterative}). Hypothesis (\ref{hyp_nf_2}) implies that condition (\ref{hyp_it_1}) holds with $\auno=1/m$. Therefore, the iterative lemma can be applied and yields
$$
\ndb{\ch{f_1}}{3-\frac{2}{m}}\leq q_1\eo\ ,\ \ \ \ndb{\ch{g_1}-\ch{\mathcal{G}}}{3-\frac{2}{m}}\leq \frac{q_1}{2}\eo+\go
$$
$$
\ntb{\ch{f_1}}{3-\frac{2}{m}}\leq q_2\Eo\ ,\ \ \ \ntb{\ch{g_1}-\ch{\mathcal{G}}}{3-\frac{2}{m}}\leq \frac{q_2}{2}\Eo+\Go\ .
$$
If $m=1$ the proof ends here.
\newline
If $m>1$, one just needs to prove that if the statement is true after $l<m$ applications of the iterative lemma, then it is also stands true after a $l+1$-th application. Thus, we suppose that after $l<m$ iterations we have 
\begin{align}\label{init}
\begin{split}
\ndb{\ch{f_l}}{3-\frac{2l}{m}} & \leq q_1^l\eo:=\eta_l\\\
\ndb{\ch{g_l}-\ch{\mathcal{G}}}{3-\frac{2l}{m}} & \leq \frac{q_1}{2}\frac{1-q_1^l}{1-q_1}\eo+\go:= \gamma_l\\
\ntb{\ch{f_l}}{3-\frac{2l}{m}} & \leq q_2^l\Eo:=\Xi_l \\
\ntb{\ch{g_l}-\ch{\mathcal{G}}}{3-\frac{2l}{m}} & \leq \frac{q_2}{2}\frac{1-q_2^l}{1-q_2}\Eo+\Go:=\Gamma_l\ .
\end{split}
\end{align}
\
\newline
Now, the aim is to apply the iterative lemma again with inequalities (\ref{init}) as initial estimates. Hypothesis (\ref{hyp_it_1}) still holds because, since $0<q_1<2/3$, one has
$$
\frac{Tm\eta_l}{2}:= q^l\frac{Tm\eo}{2}< 1\ ,
$$
so that, after having applied the iterative lemma once more, one is left with a hamiltonian in the following form:
$$
H_l:= H_0\circ\Phi_1\circ...\circ\Phi_l\circ\Phi_{l+1}=h+g_{l+1}+f_{l+1}\ ,
$$
where $\Phi_j$ is the symplectic transformation used at the $j$-th iteration of lemma (\ref{iterative}), and
$$
\{h,g_{l+1}\}=0,\ \ \ \langle f_{l+1}\rangle_h=0\ .
$$
As for the estimates on vector fields one has
$$
\ndb{\ch{f_{l+1}}}{3-\frac{2(l+1)}{m}}\leq 2\upsilon_l(m)\eta_l \ ,\ \ \ \ndb{\ch{g_{l+1}}-\ch{g_l}}{3-\frac{2(l+1)}{m}}\leq \upsilon_l(m)\eta_l 
$$
and
$$
\ntb{\ch{f_{l+1}}}{3-\frac{2(l+1)}{m}}\leq 2\Upsilon_l(m)\Xi_l \ ,\ \ \ \ntb{\ch{g_{l+1}}-\ch{g_l}}{3-\frac{2(l+1)}{m}}\leq \Upsilon_l(m)\Xi_l \ ,
$$
where the functions $\Upsilon_l$ and $\upsilon_l$ are defined exactly as in expressions (\ref{def_upsilon}) and (\ref{def_Upsilon}) by changing all the initial quantities $\eo$, $\go$, $\Eo$, $\Go$ with $\eta_l$, $\gamma_l$, $\Xi_l$, $\Gamma_l$. 
\newline
Thanks to assumption (\ref{hyp_nf_2}), one can easily check that
$$
2\upsilon_l(m)<2\upsilon_0(m)<q_1\ ,\ \ \ 2\Upsilon_l(m)<2\Upsilon_0(m)<q_2\ ,
$$
so that
$$
\ndb{\ch{f_{l+1}}}{3-\frac{2(l+1)}{m}}< q_1^{l+1}\eta_0 \ ,\ \ \ \ndb{\ch{g_{l+1}}-\ch{g_l}}{3-\frac{2(l+1)}{m}}< \frac{q_1^{l+1}}{2}\eta_0
$$
and
$$
\ntb{\ch{f_{l+1}}}{3-\frac{2(l+1)}{m}}< q_2^{l+1}\Xi_0 \ ,\ \ \ \ntb{\ch{g_{l+1}}-\ch{g_l}}{3-\frac{2(l+1)}{m}}< \frac{q_2^{l+1}}{2}\Xi_0\ .
$$
It is now easy to obtain the estimates on the resonant part of the perturbation:
\begin{align*}
\ndb{\ch{g_{l+1}}-\ch{\mathcal{G}}}{3-\frac{2(l+1)}{m}}\leq & \sum_{j=0}^{l}\ndb{\ch{g_{j+1}}-\ch{g_j}}{3-\frac{2(l+1)}{m}}  + \ndb{\ch{g_{0}}-\ch{\mathcal{G}}}{3-\frac{2(l+1)}{m}}\\
< & \sum_{j=0}^{l}\frac{q_1^{j+1}}{2}\eo+\gamma_0
= \frac{q_1}{2}\frac{1-q_1^{l+1}}{1-q_1}\eo +\go\ ,
\end{align*}
and analogously
$$
\ntb{\ch{g_{l+1}}-\ch{\mathcal{G}}}{3-\frac{2(l+1)}{m}}< \frac{q_2}{2}\frac{1-q_2^{l+1}}{1-q_2}\Eo +\Go\ .
$$
\newline
\newline
The same inductive scheme applies when calculating the size of the transformation
$$
\Psi_m:= \Phi_1\circ\Phi_2 \circ ...\circ\Phi_m\ .
$$ 
Indeed, for one application of the iterative lemma we have
$$
\ndb{\Phi_1-id}{3-\frac{2}{m}}\leq \frac{T\eo}{2}
$$
and
$$
\ntb{\Phi_1-id}{3-\frac{2}{m}}\leq \frac{T\Eo}{2}\ .
$$
In the non-trivial case $m>1$ we assume that, after $l<m$ applications, we have obtained
$$
\ndb{\Psi_l-id}{3-\frac{2l}{m}}\leq \frac{T\eo}{2}\frac{1-q_1^l}{1-q_1}
$$
and
$$
\ntb{\Psi_l-id}{3-\frac{2l}{m}}\leq \frac{T\Eo}{2}\frac{1-q_2^l}{1-q_2}\ ,
$$
where we denote
$$
\Psi_l:= \Phi_1\circ...\circ \Phi_l\ .
$$
Then, by applying lemma \ref{iterative} once more, one has
$$
\ndb{\Phi_{l+1}-id}{3-\frac{2(l+1)}{m}}\leq \frac{T\eta_l}{2}:= q_1^{l}\frac{T\eo}{2}\ ,
$$
which, in turn, implies 
\begin{align*}
\ndb{\Psi_{l+1}-id}{3-\frac{2(l+1)}{m}}
 \leq & \ndb{\Psi_{l+1}-\Psi_l}{3-\frac{2(l+1)}{m}}+ \ndb{\Psi_{l}-id}{3-\frac{2(l+1)}{m}} \\
\leq & q_1^{l}\frac{T\eo}{2} + \frac{T\eo}{2}\frac{1-q_1^l}{1-q_1}
= \frac{T\eo}{2}\frac{1-q_1^{l+1}}{1-q_1}\ .
\end{align*}
With a simliar computation one also gets 
\begin{align*}
\ntb{\Psi_{l+1}-id}{3-\frac{2(l+1)}{m}}\leq\frac{T\Eo}{2}\frac{1-q_2^{l+1}}{1-q_2}\ .
\end{align*}
As for the estimates on functions, a single application of the iterative lemma yields, by expression (\ref{bound_it_function}),
\begin{equation}
\nub{f_1}{3-\frac{2}{m}} \leq 2 \zeta_0(m) |f_0|_3,\ \ \ \nub{g_1-\mathcal{G}}{3-\frac{2}{m}} \leq \zeta_0(m)|f_0|_3+\nub{g_0-\mathcal{G}}{3}\ ,
\end{equation}
and, by hypothesis (\ref{hyp_nf_2}), this implies 
$$
\nub{f_1}{3-\frac{2}{m}} \leq p |f_0|_3,\ \ \ \nub{g_1-\mathcal{G}}{3-\frac{2}{m}} \leq \frac{p}{2}|f_0|_3+\nub{g_0-\mathcal{G}}{3}\ .
$$
Suppose, once again, that after $1\leq l<m$ iterations of lemma \ref{iterative} one has
\begin{equation}\label{f_it_2}
\nub{f_l}{3-\frac{2l}{m}} \leq p^l |f_0|_3,\ \ \ \nub{g_l-\mathcal{G}}{3-\frac{2l}{m}} \leq \frac{p}{2}\frac{1-p^l}{1-p}|f_0|_3+\nub{g_0-\mathcal{G}}{3}\ .
\end{equation}
By applying the iterative lemma \ref{iterative} one gets
\begin{equation}\label{f_it}
\nub{f_{l+1}}{3-\frac{2(l+1)}{m}}\leq 2\zeta_l(m) \nub{f_l}{3-\frac{2l}{m}}
\end{equation}
where $\zeta_l$ is defined exactly as in (\ref{def_zeta0}) by changing $\eo,\Eo,\go,\Go$ with $\eta_l,\Xi_l,\gamma_l,\Gamma_l$. By hypotheses (\ref{hyp_it_0}) one has
$$
\zeta_l(m)<\zeta_0(m)<p\ ,
$$
so that formulas (\ref{f_it_2}) and (\ref{f_it}) imply
\begin{equation}
\nub{f_{l+1}}{3-\frac{2(l+1)}{m}}\leq p \nub{f_l}{3-\frac{2l}{m}}\leq p^{l+1}\nub{f_0}{3}
\end{equation}
and
\begin{align}\label{g_it}
\begin{split}
\nub{g_{l+1}-\mathcal{G}}{3-\frac{2(l+1)}{m}}\leq & \nub{g_{l+1}-g_l}{3-\frac{2(l+1)}{m}}+ \nub{g_l-\mathcal{G}}{3-\frac{2(l+1)}{m}}\\
\leq & \frac{p}{2}^{l+1}\nub{f_0}{3} + \frac{p}{2} \frac{1-p^l}{1-p}\nub{f_0}{3}+\nub{g_0-\mathcal{G}}{3}\\
\leq & \frac{p}{2}\frac{1-p^{l+1}}{1-p}\nub{f_0}{3}+\nub{g_0-\mathcal{G}}{3}\ .
\end{split}
\end{align}
Finally, after $m$ iterations, one is left with
\begin{equation}
\nub{f_{m}}{1}\leq p^{m}\nub{f_0}{3}
\end{equation}
and 
\begin{equation}
\nub{g_{m}-\mathcal{G}}{1}\leq \frac{p}{2}\frac{1-p^{m}}{1-p}\nub{f_0}{3}+\nub{g_0-\mathcal{G}}{3}\ .
\end{equation}
\end{proof}
\ 
\newline
Together with lemmas (\ref{nf_lemma}) and (\ref{iterative}) come two important corollaries which will also be useful to prove the statement of theorem \ref{stability}. Namely, we have that the transformation $\Phi_1$ defined in the iterative lemma \ref{iterative} is invertible, as the following corollary shows.
\newtheorem{corol_it}[ac_stability]{Corollary (Single perturbative iteration)}
\begin{corol_it}\label{corol_it}
The transformation $\Phi_1$ defined in lemma (\ref{iterative}) is invertible and 
\begin{align}
\begin{split}
&\Phi_1^{-1}:\dom{1-\auno}\longrightarrow\dom{1}\\ 
\ \\
&(I,\vartheta,x,y)\longmapsto \Lambda_{\phi_1}^{-1}(I,\vartheta,x,y)\ .
\end{split}
\end{align}
Furthermore, the inverse function $\Phi_1^{-1}$ has the same size of $\Phi_1$:
\begin{align}\label{size_inv}
\begin{split}
\ndb{\Phi_1^{-1}-id}{1-\auno}\leq & \frac{T\eo}{2}\ \\
\ \\
\ntb{\Phi_1^{-1}-id}{1-\auno}\leq & \frac{T\Eo}{2}\ .
\end{split}
\end{align}
\end{corol_it}
\ 
\newline
The same result holds for the normal form transformation in lemma \ref{nf_lemma}. Namely, we have
\newtheorem{corol_nf}[ac_stability]{Corollary}
\begin{corol_nf}\label{Corollary_nf}
The transformation $\Psi_m$ defined in the normal form lemma is invertible and
\begin{align}
\begin{split}
\Psi_m^{-1}:\dom{1-\frac{T\eo}{2}\frac{1-q_1^m}{1-q_1},1-\frac{T\Eo}{2}\frac{1-q_2^m}{1-q_2}}\longrightarrow \dom{1}\\
\ \\
\Psi_m^{-1}:= \Phi_m^{-1}\circ ... \circ \Phi_1^{-1}\ ,
\end{split}
\end{align}
where $\Phi_j$ is the transformation involved at the $j$-th iteration of lemma (\ref{iterative}). 
\newline
Moreover, $\Psi_m^{-1}$ has the same size as $\Psi_m$, namely
\begin{align}\label{size_inv_2}
\begin{split}
\ndb{\Psi_m^{-1}-id}{1-\frac{T\eo}{2}\frac{1-q_1^m}{1-q_1},1-\frac{T\Eo}{2}\frac{1-q_2^m}{1-q_2}}\leq \frac{T\eo}{2}\frac{1-q_1^m}{1-q_1}\\
\ \\
\ntb{\Psi_m^{-1}-id}{1-\frac{T\eo}{2}\frac{1-q_1^m}{1-q_1},1-\frac{T\Eo}{2}\frac{1-q_2^m}{1-q_2}}\leq \frac{T\Eo}{2}\frac{1-q_2^m}{1-q_2}\ .
\end{split}
\end{align}

\end{corol_nf}
\
\newline
These two corollaries are proven in appendix \ref{Prova_corollari}.
\newline
Now, the proof of theorem (\ref{stability}) exploits a geometrical argument in order to get stability of the action variables. More precisely, variations of the projection on the line spanned by $\omega$ of the action variables are only due to the non-resonant part of the perturbation, whose magnitude has been diminished thanks to the resonant normal form developed in lemma \ref{nf_lemma}, whereas the convexity of $H_K$ bounds the diffusion in the direction orthogonal to $\omega$.  
\begin{proof}
Conditions (\ref{hyp_stab_2}) allow for the application of the normal form lemma to hamiltonian (\ref{ham_bello}). We denote the normalized coordinates with a $\tilde{\cdot}$ so that, after normalization, the hamiltonian is in the form
$$
H_m(\ti,\tth,\tx,\ty):= H\circ \Psi_m(\ti,\tth,\tx,\ty) = h(\ti)+g_m(\ti,\tth,\tx,\ty)+f_m(\ti,\tth,\tx,\ty)\ ,
$$
and estimates (\ref{ch_nf}) hold. Then, we consider the set $S_{0}(R)\times S_{0}(R)$ of initial conditions for the original non-normalized action variables $I$. Corollary \ref{Corollary_nf} insures that its image in the normalized variables is contained in the set $S_{0}(\tr)\times S_{0}(\tr)$ which in turn, by the first relation in (\ref{size_initial_rad}), is contained in the domain of the normal form. The same holds for the cartesian variables thanks to the last two inequalities in (\ref{size_initial_rad}).
\newline
We are now able to define the time of escape $\bar{t}$ of the sole action variables from the set $S_{0}(\rho+r)\times S_{0}(\rho+r)$ as the infimum time $\tau$ for which the following holds:
\begin{align}\label{tttt}
\left|\left[\Lambda^\tau_{H_m}(S_{0}(\tr)\times S_{0}(\tr)\times\mathbb{T}^2\times B_{0,0}(\xi_0)\times B_{0,0}(\xi_0))\right]^{\ti_j}\right|
\geq\rho+r
\end{align}
\newline
When considering a time $t<\bar{t}$, one can develop the flow 
$$
H_K(\ti)\circ\Lambda_{H_m}^t= [h(\ti)+\mathcal{G}(\ti)]\circ\Lambda_{H_m}^t
$$ in Taylor series with initial condition $\ti(0)\in S_{0}(\tr)\times S_{0}(\tr)$ and gets
\begin{align}
\begin{split}
H_K(\ti(t))= & H_K(\ti(0))+\left\langle\der{H_K}{\ti}(\ti(0)),\ti(t)-\ti(0)\right\rangle\\
+ & \frac{1}{2} \left(\ti(t)-\ti(0)\right)^\dag D^2 H_K(\tilde{I}^*)\left(\ti(t)-\ti(0)\right)\ ,
\end{split}
\end{align}
where $I^*$ is the point at which Lagrange's remainder is computed. 
\newline
Since the unperturbed hamiltonian $H_K$ is convex, we can write
\begin{align}\label{convexity}
\begin{split}
\left|H_K(\ti(t))-H_K(\ti(0))\right|+ & \left|\left\langle\der{H_K}{\ti}(\ti(0)),\ti(t)-\ti(0)\right\rangle\right|  \\
 \geq & \frac{\kappa}{2}|\ti(t)-\ti(0)|^2\ .
\end{split}
\end{align}
The conservation of energy
$$
H_m(\ti(t),\tth(t),\tx(t),\ty(t))=H_m(\ti(0),\tth(0),\tx(0),\ty(0))\ ,
$$
together with estimates (\ref{ch_nf}) on functions, implies the following chain of inequalities for the first term in (\ref{convexity})
\begin{align}\label{a}
\begin{split}
\left|H_K(\ti(t))-H_K(\ti(0))\right|&\leq 2\nub{g_m-\mathcal{G}}{1}+2\nub{f_m}{1}\\
&\leq 2\left[\left(\frac{p}{2}\frac{1-p^m}{1-p}+p^m\right) \nub{f_0}{3}+\nub{g_0-\mathcal{G}}{3}\ \right]\ .
\end{split}
\end{align}
\newline
On the other hand, we can split the second term in expression (\ref{convexity}) into its parallel and orthogonal component with respect to $\omega$,
\begin{align}\label{secondo_termine}
\begin{split}
\left|\left\langle \der{H_K}{I}(\ti(0)),\ti(t)-\ti(0)\right\rangle\right|
&\leq \left|\left\langle \omega,\ti(t)-\ti(0)\right\rangle\right|\\ & +\left|\left\langle \der{H_K}{I}(\ti(0))-\omega,\ti(t)-\ti(0)\right\rangle\right|\ ,\\
\end{split}
\end{align}
so that the goal now consists in bounding the two terms on the right hand side.
The former can be controlled thanks to the non-resonant nature of the exponentially small remainder $f_m$, as the following calculation shows:
\begin{align}
\begin{split}
\left|\left\langle \omega,\ti(t)-\ti(0)\right\rangle\right| = &  \nub{\left\langle\omega,\integ{0}{t}{\dot{\ti}(\tau)}{\tau}\right\rangle}{}\\
= & \nub{\integ{0}{t}{\left\langle\omega ,\der{H_m}{\tilde{\vartheta}}\right\rangle}{\tau}}{S_{0}(\tr)\times S_{0}(\tr)\times\mathbb{T}^2\times B_{0,0}(\tilde{\xi})\times B_{0,0}(\tilde{\xi})}\\
= & \nub{\integ{0}{t}{\left\langle\omega ,\der{f_m}{\tilde{\vartheta}}\right\rangle}{\tau}}{S_{0}(\tr)\times S_{0}(\tr)\times\mathbb{T}^2\times B_{0,0}(\tilde{\xi})\times B_{0,0}(\tilde{\xi})}\ .\\
\end{split}
\end{align}
Actually, in the third passage we have exploited the fact that, by taking the definition of $g_m$ into account, equality
$$
\left\langle \omega,\der{g_m}{\vartheta}\right\rangle=0
$$ 
holds.
Thus, we are left with the only contribution of the non-resonant part, which eventually yields
\begin{align}\label{b1}
\begin{split}
\left|\left\langle \omega,\ti(t)-\ti(0)\right\rangle\right| \leq &\ 	|\omega_1+\omega_2|\max_{j\in\{1,2\}}\nub{X_{f_m}^{I_j}}{1}|t| \\
\leq &\  r  |\omega_1+\omega_2| \ndb{X_{f_m}}{1}|t| \\
\leq &\ r |\omega_1+\omega_2| \eo q_1^m|t|\ \ ,\\
\end{split}
\end{align}
where, in the last inequality, we made use of estimates (\ref{ch_nf}).
\
\newline
The second term on the right-hand side of inequality (\ref{secondo_termine}) contains information about the radius of the ball of initial conditions in the normalized variables 
\begin{align}\label{b2}
\begin{split}
\left|\left\langle \der{H_K}{I}(\ti(0))-\omega,\ti(t)-\ti(0)\right\rangle\right| & =\left|\left\langle \der{H_K}{I}(\ti(0))-\omega,\ti(t)-\ti(0)\right\rangle\right|\\
& =\left|\left\langle D^2 H_K(\hat{I})(\ti(0)-I^0),\ti(t)-\ti(0)\right\rangle\right|\\
& \leq K \tilde{R}\nub{\ti(t)-\ti(0)}{S_0(\tr)\times S_0(\tr)}\ ,
\end{split}
\end{align}
where $\hat{I}\in S_{0}(\tilde{R})$ is the point associated to the remainder in Lagrange form.
\newline
By plugging (\ref{a}), (\ref{b1}) and (\ref{b2}) into (\ref{convexity}) we have 
\begin{align}\label{convexity_2}
\begin{split}
& 2\left[\left(\frac{p}{2}\frac{1-p^m}{1-p}+p^m\right) \nub{f_0}{3}+\nub{g_0-\mathcal{G}}{3}\ \right]\\
& + r |\omega_1+\omega_2| \eo q_1^m|t| 
+ K\tilde{R} \nub{\left(\ti(t)-\ti(0)\right)}{S_0(\tr)\times S_0(\tr)} \\
&\geq \frac{\kappa}{2}|\ti(t)-\ti(0)|^2_{S_0(\tr)\times S_0(\tr)}\ 
\end{split}
\end{align}
whose solution is 
\begin{align}\label{sol}
\begin{split}
0\leq \nub{\ti(t)-\ti(0)}{S_0(\tr)\times S_0(\tr)} \leq
\frac{K}{\kappa} \tr +\sqrt{\left(\frac{K}{\kappa} \tr\right)^2+a(t)}\ ,
\end{split}
\end{align}
where we denote with $a(t)$ the quantity
\begin{align}\label{definition_a}
\frac{2}{\kappa}\left[\left(p\frac{1-p^m}{1-p}+2p^m\right) \nub{f_0}{3}+2\nub{g_0-\mathcal{G}}{3}\ + r |\omega_1+\omega_2| \eo q_1^m|t|\right]\ .
\end{align}
In the final part of the proof, an explicit estimate on the escape time $\bar{t}$ will be found. Indeed, for every time $t<\bar{t}$ one has 
\begin{align}
\begin{split}\label{itilde2}
\nub{\tilde{I}(t)}{S_0(\tr)\times S_0(\tr)}\leq\nub{\tilde{I}(t)-\ti(0)}{S_0(\tr)\times S_0(\tr)}+\nub{\tilde{I}(0)}{S_0(\tr)\times S_0(\tr)}\leq\rho+r\ .
\end{split}
\end{align}
With the help of inequality (\ref{sol}) and by taking the definition of $\tr$ into account, the latter inequality can be rewritten as
\begin{equation}\label{itilde}
\left(\frac{K}{\kappa}+1\right) \tr +\sqrt{\left(\frac{K}{\kappa} \tr\right)^2+a(t)}\leq \rho+r\ .
\end{equation}
Extracting $t$ from the above formula one is left with
\begin{equation}\label{bound_t}
t<\frac{C_1(R)}{C_2}\ q_1^{-m}\ ,
\end{equation}
where the constants read
\begin{align}\label{Const}
\begin{split}
C_1(R) :=  \frac{\kappa}{2} & \left\{\left[\rho+r-\left(\frac{K}{\kappa}+1\right)\tr\right]^2 -\left(\frac{K}{\kappa}\tr\right)^2 \right\}  \\
  - &\left(p\frac{1-p^m}{1-p}+2p^m\right) \nub{f_0}{3}-2\nub{g_0-\mathcal{G}}{3}\\
C_2 := & r|\omega_1+\omega_2|\eo\ .
\end{split}
\end{align}
When coming back to the original, non-resonant variables, one must add to the variation calculated in (\ref{sol}) the size of the normal form transformation; this eventually yields
\begin{equation}
\nub{I(t)-I(0)}{S_0(R)}\leq \frac{K}{\kappa} \tr +\sqrt{\left(\frac{K}{\kappa} \tr\right)^2+a(t)}+\frac{T\eo}{2}\frac{1-q_1^m}{1-q_1}r\ ,
\end{equation}
so that the theorem is proved.
\end{proof}
\subsubsection{Confinement of the eccentricities}\label{conf_ecc}
In this section we prove the second part of theorem \ref{ac_stability} and insure that the diffusion of the cartesian variables is bounded thanks to the conservation of the angular momentum. In particular we have the following
\newtheorem{eccentricities}[ac_stability]{Theorem (Stability of the cartesian variables)}
\begin{eccentricities}\label{eccentricities}
Assume the hypotheses and the notations of Theorem \ref{stability}. Consider, in particular, the domain $\mathcal{D}_{\rho,4r,4s,\xi,4u}$ of analyticity for hamiltonian (\ref{ham_bello}). Choose a radius of initial conditions 
$$
\xi_0=\displaystyle\max_{j\in\{1,2\}}\{x_j(0)^2+y_j(0)^2\}
$$ 
for the cartesian variables and suppose that the size $\xi$ of the real domain, together with the analyticity width $u$, satisfies
\begin{align}\label{cond_u}
\begin{split}
\xi+\left(1-\frac{T\Eo}{2}\frac{1-q_2^m}{1-q_2}\right)u & >\xi_0\\
\xi+\left(1-\frac{T\Eo}{2}\frac{1-q_2^m}{1-q_2}\right)u & \geq \sqrt{\Lambda_1^0+\Lambda_2^0+2R_f(\bar{t})-\mathcal{N}^-(\xi_0)}\ ,
\end{split}
\end{align}
where $\mathcal{N}^-(\xi_0)$ denotes the minimal value that the angular momentum can take, 
\begin{equation}
\mathcal{N}^-(\xi_0):=(\Lambda_1^0-R)\sqrt{1-\bar{e}_{1}(0,\xi_0)^2}+(\Lambda_2^0-R)\sqrt{1-\bar{e}_{2}(0,\xi_0)^2}\ ,
\end{equation}
and $\bar{e}_{1}(0,\xi_0),\bar{e}_{2}(0,\xi_0)$ are the maximal initial eccentricities which are compatible with $\xi_0$, namely, by expression (\ref{poincare}),
\begin{equation}\label{max_ecc}
\bar{e}_j(0,\xi_0)=\sqrt{1-\left(1-\frac{\xi_0^2}{2(\Lambda^0_1-R)}\right)^2}\ ,\ \ j\in\{1,2\}.
\end{equation}
\newline
Then there exist two functions $(\overline{e}_1,\overline{e}_2)\in\ ]0,1[\times]0,1[$  depending on time and on $\xi_0$ such that, for all $t<\bar{t}$, one has
\begin{equation}
e_1(t)<\overline{e}_1(t,\xi_0)\ ,\ \ \ e_2(t)<\overline{e}_2(t,\xi_0)\ .
\end{equation} 
Moreover, $\overline{e}_1$ and $\overline{e}_2$ can be explicitly computed and read 
\begin{align}
\begin{split}
\overline{e}_1: (t,\xi_0)\longmapsto & \sqrt{1-\left(\frac{\mathcal{N}^{-}(\xi_0)-\Lambda_2^0-R_f(t)}{\Lambda_1^0+R_f(t)}\right)^2}\ ,\\
\overline{e}_2:  (t,\xi_0)\longmapsto & \sqrt{1-\left(\frac{\mathcal{N}^{-}(\xi_0)-\Lambda_1^0-R_f(t)}{\Lambda_2^0+R_f(t)}\right)^2}\ .
\end{split}
\end{align}
\end{eccentricities}
\begin{proof}
We define the time of escape of the cartesian variables from the domain of the normal form, namely
\begin{equation}
t_e:=\inf\left\{t\in\mathbb{R}:\left[\xi+\left(1-\frac{T\Eo}{2}\frac{1-q_2^m}{1-q_2}\right)u\right]^2 \leq \max_{j\in\{1,2\}}\{x_j(t)^2+y_j(t)^2\}\right\}
\end{equation}
and we have that for all $t<t_e$ the non-normalized cartesian variables are in the image of the normal form transformation described in lemma \ref{nf_lemma}.
\newline
Indeed, by the very definitions (\ref{poincare}) of $x_j$ and $y_j$ and by expression (\ref{angular_momentum}), we can write the above inequality in the form
\begin{equation}\label{condizione}
\left[\xi+\left(1-\frac{T\Eo}{2}\frac{1-q_2^m}{1-q_2}\right)u\right]^2 > 2(\Lambda_1(t)+\Lambda_2(t)-\mathcal{N})\ \ \ \ \forall t<t_e\ 
\end{equation}
and one immediately sees that hypothesis (\ref{cond_u}), together with theorem \ref{stability}, implies that $\bar{t}<t_e$. As a matter of fact, we now have that for any initial condition $(x_j(0),y_j(0)),j\in\{1,2\},$ in the original non-normalized cartesian variables such that 
$$
x_j^2(0)+y_j^2(0)<\xi_0\ ,
$$
with $\xi_0$ satisfying (\ref{cond_u}), one is insured that the system does not escape from the domain of analyticity for any time $t$ inferior to the time of confinement in the action variables. 
\newline
Moreover, since $\mathcal{N}$ is an integral of motion, we have that for all times $t\in\mathbb{R}$
\begin{equation}
\mathcal{N}\geq \mathcal{N}^-(\xi_0)\ ;
\end{equation}
solving equation (\ref{angular_momentum}) with respect to $e_1(t)$ yields
\begin{equation}
e_1(t)=\sqrt{1-\left(\frac{\mathcal{N}-\Lambda_2(t)\sqrt{1-e_2(t)^2}}{\Lambda_1(t)}\right)^2}
\end{equation}
so that the worst case scenario corresponds to $\mathcal{N}=\mathcal{N}^-(\xi_0)$ and $e_2(t)=0$. Thus, we can say that for all $t<\bar{t}$
\begin{equation}
e_1(t)\leq \overline{e}_1 (t,\xi_0):=\sqrt{1-\left(\frac{\mathcal{N}^{-}(\xi_0)-\Lambda_2^0-R_f(t)}{\Lambda_1^0+R_f(t)}\right)^2}\ ,
\end{equation}
where, once again, we have used Theorem $\ref{stability}$ to give an upper bound to the actions.
With a similar calculation one gets the expression for $\overline{e}_2 (t,\xi_0)$.
\end{proof}
\subsubsection{Proof of the main stability theorem}
Theorems \ref{stability} and \ref{eccentricities} together imply theorem \ref{ac_stability}. Such result is strictly local since it has been constructed in the neighborhood of a periodic torus. In order to obtain a global result (which is not our purpose here), one could make use of Dirichlet theorem so to cover the whole phase space with periodic orbits of the unperturbed system, as in \cite{Lochak_Neishtadt_Niederman_1994}. 

\section{The restricted, circular, planar three-body problem}\label{restricted_tbp}
\subsection{Motivation}\label{motivation}
Theorem \ref{ac_stability} insures Nekhoroshev-like stability for the plane, planetary three-body problem in the neighborhood of a periodic orbit of the unperturbed system. Clearly, the method we used to prove it can be applied to any quasi-integrable system, provided that one explicitly knows the analyticity widths and the initial bounds on its hamiltonian vector fields. In the previous section, we just had information on the size of the perturbation in its domain of analyticity, so that we were obliged to make use of the Cauchy inequalities in order to get estimates (\ref{useful}). These inequalities, in turn, are derived from the well-known Cauchy representation formula (see e.g. \cite{Schiedemann_2005}) with the help of generic bounds that may not be sharp at all in many concrete applications. Therefore, a direct computation of the derivatives, when possible, may lead to improved initial estimates. This turns out to be very important in the case we are considering since any initial gain in the estimates for functions and vector fields grows exponentially in the number of iterations of lemma \ref{iterative}, as theorem \ref{stability} shows. \newline 
Moreover, at least in principle, theorem \ref{stability} may be limited in its physical applications by the complex singularities of the considered hamiltonian. Indeed, as we shall see when considering numerical computations in paragraph \ref{num_comp}, the value of the analyticity width $r$ in the action variables which yields the longest time of stability increases with the size $\varepsilon$ of the perturbation. Thus, at least in principle, singularities may be encountered when considering a domain which is too large in the action variables. Knowing exactly where these singularities are in complex action-angle coordinates turns out to be a very diffucult matter when considering problems in celestial mechanics. In \cite{Castan_2017}, for example, one is given sufficient conditions so to avoid them. \newline
In order to see what happens when such difficulties can be overcome, it is interesting to apply the results of section \ref{tbpp} to a system whose hamiltonian vector fields can be directly estimated without making use of the Cauchy inequalities and whose hamiltonian perturbation has no complex singularities. In this spirit, we chose to investigate the Nekhoroshev-like stability in the neighborhood of a periodic torus for the restricted, circular, planar three-body problem as modeled in \cite{Celletti_Chierchia_2007} and \cite{Celletti_Ferrara_1996}.

\subsection{Hamiltonian framework}\label{ham_fr_r}
Here, we briefly recall the hamiltonian setup stated in \cite{Celletti_Chierchia_2007} and we give some suitable definitions. Consider, once again, three coplanar bodies mutually interacting through the sole gravitational force and label them with an index $j\in\{0,1,2\}$. In this case we suppose that the mass $m_0$ is much greater than $m_1$ and that $m_2=0$. When considering heliocentric coordinates, we are left with an elliptic orbit of frequency $\omega_g$ and semi-major axis $a_1$ for body $1$ around body $0$ and with body $2$ undergoing interactions with the primaries. The circular approximation consists in assuming a null eccentricity for the trajectory of body $1$ in the configuration space. In this framework, suitable action-angle coordinates for body $2$, expressed as functions of its time-dependent orbital elements, are
\begin{equation}\label{delaunay_restricted}
\begin{cases}
L\ :=\ & \mu\sqrt{G_N m_0 a}\\
G\ :=\ & L\sqrt{1-e^2}\\
l\ :=\ & \lambda\\
g\ :=\ & \gamma-\tau\ \\
\end{cases}\ ,
\end{equation}
where we have denoted $\mu:=(G_N m_0)^{-2/3}$ and where $\lambda,\gamma$ respectively stand for the mean longitude and the argument of periapsis for body $2$ and $\tau$ is  the mean longitude of body $1$.
\begin{figure}[h]
\centering
\includegraphics[scale=0.3]{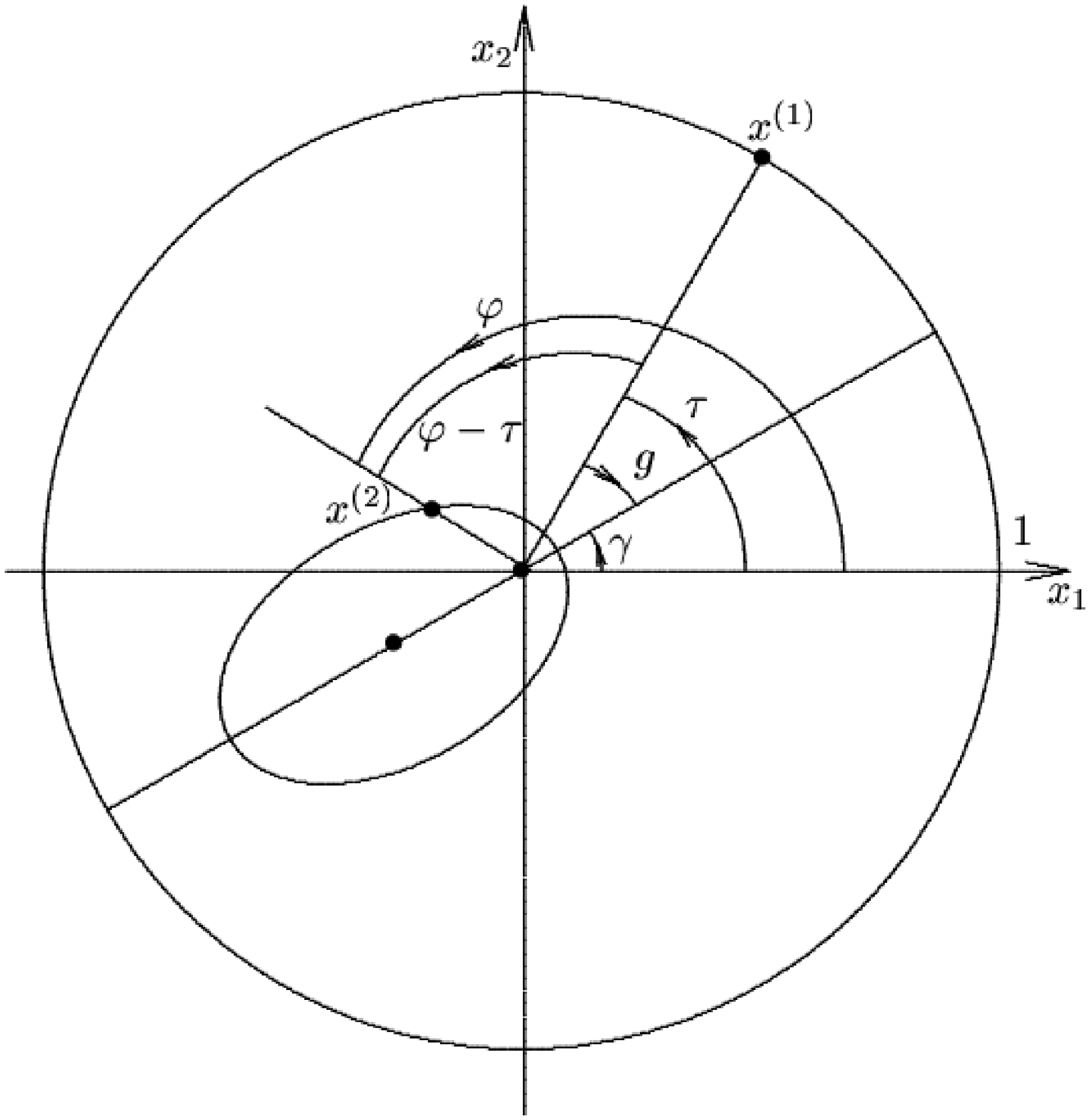}\label{cc}
\caption{Orbital elements following the construction in \citep{Celletti_Chierchia_2007}.}
\end{figure}
Following the construction in \cite{Celletti_Chierchia_2007}, the motion of body $2$ is governed by the following hamiltonian:
\begin{equation}
H(L,G,l,g):=H_0(L,G)+\varepsilon H_1(L,G,l,g)\ ,
\end{equation}
where
\begin{equation}\label{unperturbed_restricted}
H_0(L,G):= -\frac{1}{2L^2}-\omega_g G\ ,
\end{equation}
and $H_1$ is a trigonometric polynomial which is obtained by retaining only the most relevant harmonics from the Fourier expansion of the complete perturbation. A rigorous criterion insuring that the truncated model stays close to the complete one is implemented in \cite{Celletti_Chierchia_2007}. \newline
Since we are interested in the behaviour of this system in the neighbourhood of a $p:q$ resonance corresponding to a $T$-periodic torus, we can consider the same resonant decomposition that held for the planetary three-body problem in section \ref{tbpp}. For the sake of simplicity, we shall use the same symbols to denote quantities that play the same roles in the two cases. Thus, we are allowed to write
\begin{equation}\label{restricted_ham}
H(L,G,l,g):=h(L,G)+g_0(L,G,l,g)+f_0(L,G,l,g)\ ,
\end{equation}
where $h$ generates the integrable linear flow of frequencies $(\omega_l,\omega_g)$, and $g_0,f_0$ are the resonant and non resonant perturbations. In this case, $g_0$ and $f_0$ are two trigonometric polynomials. Moreover, as we did in the prequel, we use the symbol $\mathcal{G}$ to denote the remainder of order 2 in the expansion of $H_0$ and $(L^0,G^0)$ to denote the action variables corresponding to the exact resonance for the integrable hamiltonian.
\newline
After these observations, we now consider the domain
\begin{align}\label{dom_res}
\begin{split}
\dom{\rho_L,\rho_G,r_L,r_G,s_l,s_g}:=\{& (L,G,l,g) \in\mathbb{C}^4:\\
& \exists \ L^*\in S_{L^0}(\rho_L) \text{ such that } \left|L-L^*\right|<r_L\ ,\\
& \exists \ G^*\in S_{G^0}(\rho_G) \text{ such that } \left|G-G^*\right|<r_G\ ,\\
& \Re e(l,g)\in\mathbb{T}^2\ ,\ \  
\left|\Im m(l)\right|<s_l\ ,\ \ \left|\Im m(g)\right|<s_g
\}
\end{split}
\end{align}
with the same shorthand notations we defined in (\ref{shorthands}). Remark that the values for the analiticity widths can be arbitrary in this case since there are no complex singularities. Then, we assume that the truncated model described by hamiltonian (\ref{restricted_ham}) satisfies the same assumptions on the magnitude of the discarded harmonics as in \cite{Celletti_Chierchia_2007}. Such condition was always checked when performing the computations of section \ref{num_comp}. In this spirit, we introduce the following definition:
\newtheorem{norms2}[sets]{Definition}
\begin{norms2}
For $(j,\sigma_j)\in \{(L,r_L),(G,r_G),(l,s_l),(g,s_g)\}$, for any open set $\mathcal{E}\subset\mathbb{C}^4$ and for any continuous, bounded vector field $v: \mathcal{E}\longrightarrow \mathbb{C}^4$, we define the following norm for each component $v^j$:
\begin{align}\label{norme_anisotrope}
\begin{split}
\ndbs{v}{\mathcal{E}}{j}:= \frac{\nub{v^j}{\mathcal{E}}}{\sigma_j}\ .\\
\end{split}
\end{align}
\end{norms2}
\
\newline
As we did in section \ref{tbpp}, we also assume the following bounds on the anisotropic norms
\begin{align}\label{bornes}
\begin{split}
\ndbs{\ch{\mathcal{G}}}{3}{L}\leq \delta,\ \ndbs{\ch{f_0}}{3}{j}\leq \eta_0^j,\ \ndbs{\ch{g_0-\mathcal{G}}}{3}{j}\leq \gamma_0^j,\ \ \ \ j\in\{L,G,l,g\}
\ .
\end{split}
\end{align}
Notice that $\mathcal{G}$ only depends on the first action $L$ as $H_0(L,G)$ is linear with respect to $G$. 
\newline
Since perturbation $H_1$ is an explicit finite sum of Fourier harmonics, quantities (\ref{bornes}) can be estimated without making use of the Cauchy inequalities. 
As in the planetary case, the non-null eigenvalue of the hessian matrix $D^2 H_0(I)$, denoted $\varrho(L)$, satisfies 
$$
\kappa \leq |\varrho(L)| \leq K
$$ 
for all values of $L$ in the domain $\dom{\rho_L,\rho_G,r_L,r_G,s_l,s_g}$, where $K$ and $\kappa$ are two positive constants. As in the planetary case, both quantities can be explicitly computed. Finally, we introduce five real functions that play the same role that (\ref{def_upsilon}), (\ref{def_Upsilon}) and (\ref{def_zeta0}) played in the planetary case, namely
\begin{align}
\begin{split}
\upsilon_0^L: x\longmapsto & \frac{\left(Tx\right)^2\eol}{2}\chi_0+\frac{Tx}{2}\chi_0+\left(1+\frac{\goL}{\eoL}\right)x\Theta_0\\
+ &\frac{s_2}{s_1}\frac{r_2}{r_1}\frac{Tx^2}{2\eoL}\left\{T\eog\eoG\chi_0+\left[\eog(\eoG+\goG)+\eoG(\eog+\gog)\right]\Theta_0\right\}\\
+& \frac{Tx^2}{2}\left[\eol\left(1+\frac{\goL}{\eoL}\right)+\eol+\gol+\dd\right]\Theta_0\ ,\\
\end{split}
\end{align}
\begin{align}
\begin{split}
\upsilon_0^G: x\longmapsto & \frac{\left(Tx\right)^2\eog}{2}\chi_0+\frac{Tx}{2}\chi_0+\left(1+\frac{\goG}{\eoG}\right)x\Theta_0\\
+ &\frac{s_1}{s_2}\frac{r_1}{r_2}\frac{Tx^2}{2\eoG}\left\{T\eol\eoL\chi_0+\left[\eol(\eoL+\goL)+\eoL(\eol+\gol+\dd)\right]\Theta_0\right\}\\
+& \frac{Tx^2}{2}\left[\eog\left(1+\frac{\goG}{\eoG}\right)+\eog+\gog\right]\Theta_0\ ,\\
\end{split}
\end{align}
\begin{align}
\begin{split}
\upsilon_0^l: x\longmapsto & \left(Tx\right)^2\frac{\eoL}{2}\chi_0+\frac{Tx}{2}\chi_0+\left(1+\frac{\gol}{\eol}+\frac{\dd}{\eol}\right)x\Theta_0\\
+ &\frac{s_2}{s_1}\frac{r_2}{r_1}\frac{Tx^2}{2\eol}\left\{T\eog\eoG\chi_0+\left[\eog(\eoG+\goG)+\eoG(\eog+\gog)\right]\Theta_0\right\}\\
+& \frac{Tx^2}{2}\left[\eoL\left(1+\frac{\gol}{\eol}+\frac{\dd}{\eol}\right)+\eoL+\goL\right]\Theta_0\ ,\\
\end{split}
\end{align}
\begin{align}
\begin{split}
\upsilon_0^g: x\longmapsto & \left(Tx\right)^2\frac{\eoG}{2}\chi_0+\frac{Tx}{2}\chi+\left(1+\frac{\gog}{\eog}\right)x\Theta_0\\
+ &\frac{s_1}{s_2}\frac{r_1}{r_2}\frac{Tx^2}{2\eog}\left\{T\eol\eoL\chi_0+\left[\eol(\eoL+\goL)+\eoL(\eol+\gol+\dd)\right]\Theta_0\right\}\\
+& \frac{Tx^2}{2}\left[\eoG\left(1+\frac{\gog}{\eog}\right)+\eoG+\goG\right]\Theta_0\ ,\\
\\
\zeta_0: x\longmapsto & \frac{Tx}{2}\chi_0\ ,\\
\end{split}
\end{align}
where we have set 
\begin{equation}
\chi_0:= \sup\{\eoL+\goL,\ \eoG+\goG,\ \eol+\gol+\dd,\ \eog+\gog\}\ ,
\end{equation}
\begin{equation}
\Theta_0:= \frac{T}{2}\sup\{\eoL,\eoG,\eol,\eog\}\ .
\end{equation}
\subsection{Stability in the neighbourhood of a periodic torus}
Taking the definitions of the previous paragraph into account, we are now ready to state a stability result for the restricted problem. Since hamiltonian (\ref{unperturbed_restricted}) is strictly convex only in the $L$ coordinate, the method we used when proving theorem \ref{stability} can only be used to confine this variable as the following theorem shows. The $G$ variable could be bounded by making use of some arguments exploiting quasi-convexity (see e.g. \cite{Lochak_1992}). However, since we are in the particular case of a two degrees of freedom system, we chose to confine the $G$ variable by making use of the conservation of energy since such approach involves simpler calculations. 
\newtheorem{stability2}[ac_stability]{Theorem (Stability for the whole system)}
\begin{stability2}\label{stability2}
Assume the constructions above for hamiltonian (\ref{restricted_ham}) in $\dom{\rho_L,\rho_G,r_L,r_G,s_l,s_g}$. Suppose that there exist $m\in\mathbb{N}$ and five numbers $p,q_j\in\left]0,\displaystyle\frac{2}{3}\right[$, where $j\in\{L,G,l,g\}$ is an alphabetical index, satisfying
\begin{align}\label{hyp_stab_3}
2\upsilon_0^j(m)<q_j\ ,\ \ \ 
2\zeta_0(m)< p\ .
\end{align}
Fix $\varepsilon$ sufficiently small and suppose that the analiticity radii $r_G,\rho_G$ are sufficiently big so that one can pick two positive real numbers $L_{init},G_{init}$ satisfying
\begin{align}\label{size_initial_rad_2}
\begin{split}
C_3(L_{init})&>0\ ,\\
\end{split}
\end{align}
and
\begin{equation}\label{GGG}
|G^0+G_{init}|+\frac{1}{\omega_G}\left(W(L_{init})+2\varepsilon|H_1|_1\right)\leq \rho_G+r_G-\frac{T\eo^G}{2}\frac{1-q_G^m}{1-q_G}r_G\ ,
\end{equation}
where 
\begin{align}
\begin{split}
C_3(L_{init}) := & \frac{\kappa}{2}  \left\{\left[\rho_L+r_L-\left(\frac{K}{\kappa}+1\right)L_{init}\right]^2 -\left(\frac{K}{\kappa}L_{init}\right)^2 \right\} \\
    &-\left(p\frac{1-p^m}{1-p}+2p^m\right) \nub{f_0}{3}-2\nub{g_0-\mathcal{G}}{3}\ ,\\
    \ \\
W(L_{init}):= & \frac{\left(L_{init}+V(\rho_L,r_L,\eo^L)\right)\left(L_{init}+2L^0+V(\rho_L,r_L,\eo^L)\right)}{2\left(L^0-V(\rho_L,r_L,\eo^L)\right)^2(L^0-L_{init})^2}
\end{split}
\end{align}
and we have denoted
\begin{equation}
V(\rho_L,r_L,\eo^L)=\rho_L+r_L-\frac{T\eo^L}{2}\frac{1-q_L^m}{1-q_L}r_L\ .
\end{equation}
\ 
\newline
Then there exist a positive constant $C_4$ and three functions $L_f,A_{\pm}:\mathbb{R}\longrightarrow\mathbb{R}$ such that, for any initial condition
\begin{equation}
(L(0),G(0))\in S_{L^0}\left(L_{init.}\right)\times S_{G^0}\left(G_{init.}\right)
\end{equation}
and for any time 
\begin{equation}
|t| < \bar{t}:= \frac{C_3(L_{init})}{C_4}q_L^{-m}\ ,
\end{equation}
the flow of $H$ stays inside $\mathcal{D}_{1-\frac{T\eo^j}{2}\frac{1-q_j^m}{1-q_j}}\ ,j\in\{L,G,l,g\},$ and one has
\begin{align}\label{Lf}
\begin{split}
\nub{L(t)-L(0)}{S_{L^0}(L_{init})} \leq L_f(t)\ , \\ 
\end{split}
\end{align}
whereas the eccentricity is bounded by
\begin{equation}\label{be}
\sqrt{1-A_+(t)}\leq e(t)\leq \sqrt{1-A_-(t)}\ .
\end{equation}
Moreover, explicit expressions for such constant and functions can be found and read:
\begin{align}\label{C3}
\begin{split}
C_4 := &  \left|\omega_l \eo^L r_L + \left(\frac{q_G}{q_L}\right)^m \omega_g \eo^G r_G\right|\ ,\\
\ \\
L_f : &  \ t\longmapsto  \frac{K}{\kappa} \tilde{L} +\sqrt{\left(\frac{K}{\kappa} \tilde{L}\right)^2+b(t)}+\frac{T\eo^L}{2}\frac{1-q_L^m}{1-q_L}r_L\ ,\\
\ \\
A_\pm(t):= & \frac{1}{a(t)}\left[a(0)y(0)+\frac{B^2(t)}{G_N m_0(\mu\omega_G)^2}\right]\pm\frac{2B(t)}{a(t)\mu\omega_G}\sqrt{\frac{a(0)y(0)}{G_N m_0}}\ ,
\end{split}
\end{align}
where we have denoted  
\begin{align}
\begin{split}
\tilde{L} := &  L_{init.}+\frac{T\eo^L}{2}\frac{1-q_L^m}{1-q_L}r_L
\ \\
b(t):= & \frac{2}{\kappa}\left[\left(p\frac{1-p^m}{1-p}+2p^m\right) \nub{f_0}{3}+2\nub{g_0-\mathcal{G}}{3}\ + C_4q_L^m|t|\right]\\
y(0):= & \sqrt{1-e^2(0)}\\
B(t):= & \left|\frac{1}{2L^2(t)}-\frac{1}{2L^2(0)}\right|
+\varepsilon\left| H_1 \circ \Lambda_H^t-H_1 \circ \Lambda_H^0\right|
\end{split}\ .
\end{align}
\end{stability2}
\
\newline
\begin{proof}
The stability of the $L$ coordinate is demonstrated by putting the non-resonant perturbation into normal form and by applying exactly the same geometrical argument of theorem \ref{stability}. Clearly, two lemmas corresponding to lemmas \ref{nf_lemma} and \ref{iterative} in section \ref{stab} hold also in this case: their statements and proofs can be found in appendix \ref{lem_res}.
\newline
As for the bound on the $G$ variable, we exploit the conservation of energy for hamiltonian (\ref{restricted_ham}),
\begin{equation}
H(L(t),G(t),l(t),g(t))=H(L(0),G(0),l(0),g(0))\ ,
\end{equation}
which yields the following bound:
\begin{equation}
|G(t)-G(0)|\leq \frac{1}{\omega_G}\left(\left|\frac{1}{2L^2(t)}-\frac{1}{2L^2(0)}\right|+\varepsilon|H_1 \circ \Lambda_H^t-H_1 \circ \Lambda_H^0|\right)\ ,
\end{equation}
and one sees that, thanks to hypothesis (\ref{GGG}) and with the help of standard bounds, such inequality insures that the variable $G$ stays in the considered domain for any time $t$ inferior to the time of stability of the $L$ variable.
By taking the second expression in (\ref{delaunay_restricted}) into account and solving with respect to $e$ one gets inequality (\ref{be}). Moreover, by considering the expression for (\ref{Lf}), one obtains a suitable supremum for the eccentricity. 
\end{proof}
\section{Examples and concrete computations}\label{num_comp}
In the last part of this work, we have performed computations in order to investigate the mechanisms leading to Nekhoroshev stability for some astronomical systems close to resonances. This also allows for a disentanglement of the limits that such techniques can encounter and suggest solutions on how to overcome them. In particular, as we shall show in the sequel, various obstacles may arise when increasing the size $\varepsilon$ of the perturbation. However, there seems to be good hopes of reaching physical values for $\varepsilon$, at least in the truncated, restricted, circular, planar three-body problem. Moreover, good thresholds on the size of the perturbation were reached both in the KAM framework (see \cite{Celletti_Chierchia_2007}) and in the Nekhoroshev one (see \cite{Celletti_Ferrara_1996}) when considering the latter model in other regions of the phase space. The computations that we present hereafter were carried out with the help of codes written in Mathematica language.  
\subsection{The 5:2 resonance for the planetary problem}
It is known since a long time (see e.g. \cite{Goldreich_1965}) that various commensurability relations hold for the frequencies of celestial bodies in the Solar System. For example, Jupiter and Saturn lie very close to the $5:2$ mean-motion resonance (see \cite{Michtchenko_Ferraz-Mello_2001} and references therein for an astronomical point of view on this phenomenon) and the ratio of their masses is close to $10^{-3}$.
Moreover, the relative inclinations of their orbital planes are small. 
\newline
In this spirit, we choose to study the plane, planetary three-body problem described in section \ref{tbpp} with explicit values corresponding to a Sun-Jupiter-Saturn model (with smaller masses) in 5:2 resonance. The initial data for the eccentricities and for the resonant action $\Lambda_1^0$ are set to be those of J2000 (see \url{https://nssdc.gsfc.nasa.gov/planetary/factsheet/}), whereas $\Lambda_2^0$ is determined by the resonant relation between the two mean motion frequencies and by Kepler's third law. Then, for different initial conditions in the actions in a neighborhood of $(\Lambda_1^0,\Lambda_2^0)$ and for different values of $\varepsilon$, we compute by trial and error the analyticity widths and the number $m$ of iterations of lemma \ref{iterative} which yield the longest times of stability $\bar{t}$. The magnitude of the perturbing function on the chosen domain of analyticity was estimated with the help of majorant series thanks to a code provided by Dr. Thibaut Castan (see \cite{Castan_2017} for more details). The best results are obtained for $R=\rho=0$, which amounts to setting the initial conditions in the action variables exactly at the resonance $(\Lambda_1(0)=\Lambda_1^0,\Lambda_2(0)=\Lambda_2^0)$. For other initial conditions in the actions variables not exactly at the resonance, one obtains times of stability which are comparable with the age of the Solar System for similar magnitudes of the perturbation, provided that the radius of initial conditions satisfies $R\lesssim 8\times 10^{-7}\times \max\{\Lambda_1^0,\Lambda_2^0\}$ and that $\rho\lesssim 1\times10^{-6}\times \max\{\Lambda_1^0,\Lambda_2^0\}$. Such results are contained in the tables below.
\begin{center}\label{Table1}
  \begin{tabular}{ | c | c | c | c | c | c | }
    \hline
    $\log(\varepsilon)$ & $m$ & $\bar{t}\ (y)$ & $R_f(\bar{t})/\max\{\Lambda_1^0,\Lambda_2^0\}$ & $\bar{e}_1$ & $\bar{e}_2$\\ \hline \hline
    $-12.25$ & $61$ & $5.71\times 10^{39}$ & $7.07\times 10^{-7}$ & $0.0595$ & $0.0932$\\ \hline
    $-12.00$ & $45$ & $1.17\times 10^{29}$ & $9.65\times 10^{-7}$ & $0.0595$ & $0.0932$ \\ \hline
    $-11.75$ & $34$ & $1.25\times 10^{21}$ & $1.30\times 10^{-6}$ & $0.0595$ & $0.0932$ \\ \hline
    $-11.50$ & $25$ & $1.48\times 10^{15}$ & $1.80\times 10^{-6}$ & $0.0595$ & $0.0933$ \\ \hline
    $-11.25$ & $18$ & $5.75\times 10^{10}$  & $2.51\times 10^{-6}$ & $0.0595$ & $0.0933$ \\ \hline
    $-11.00$ & $14$ & $3.08\times 10^{7}$ & $3.54\times 10^{-6}$ & $0.0596$ & $0.0934$ \\ \hline
    $-10.75$ & $10$  & $1.22\times 10^{5}$ & $5.14\times 10^{-6}$  & $0.0596$ & $0.0934$ \\ \hline
  \end{tabular}
 \captionof{table}{{\footnotesize From left to right: magnitude of the perturbation, number of iterative steps, time of stability, radius of confinement in the actions, maximal values for the eccentricities. Initial conditions in the actions are supposed to be those corresponding exactly to the 5:2 resonance, whereas the initial values for the eccentricities are set to be those for Jupiter and Saturn at J2000.}}
\end{center}

\begin{center}
  \begin{tabular}{ | c | c | c | c | c | }
    \hline
   $\log(\varepsilon)$ & $r/\max\{\Lambda_1^0,\Lambda_2^0\}$ & $s$ & $|1-\beta|$ & $\xi$ \\ \hline \hline
    $-12.25$ & $3.54\times 10^{-7}$ & $3.97\times 10^{-2}$ & $\sim 6\times 10^{-4}$ & $4.36\times 10^{15}$ \\ \hline
    $-12.00$ & $4.83\times 10^{-7}$ & $3.95\times 10^{-2}$ & $\sim 4\times 10^{-4}$ & $5.80\times 10^{15}$  \\ \hline
    $-11.75$ & $6.51\times 10^{-7}$ & $3.91 \times 10^{-2}$ & $\sim 2\times 10^{-2}$ & $7.73\times 10^{15}$ \\ \hline
    $-11.50$ & $9.04\times 10^{-7}$ & $3.89 \times 10^{-2}$ & $\sim 5\times 10^{-4}$ & $1.03\times 10^{16}$\\ \hline
    $-11.25$ & $1.25\times 10^{-6}$ & $3.85 \times 10^{-2}$ & $\sim 3 \times 10^{-5}$ & $1.37\times 10^{16}$ \\ \hline
    $-11.00$ & $1.76\times 10^{-6}$ & $3.82 \times 10^{-2}$ & $\sim 7 \times 10^{-4}$ & $1.83\times 10^{16}$ \\ \hline
    $-10.75$ & $2.57\times 10^{-6}$ & $3.76\times 10^{-2}$ & $\sim 8\times 10^{-5}$ & $2.43\times 10^{16}$\\ \hline
  \end{tabular}
\captionof{table}{{\footnotesize From left to right: magnitude of the perturbation, analyticity widths for the action-angle variables and for the cartesian coordinates. Initial conditions are the same of Table 1.}}
\end{center}
\begin{center}
  \begin{tabular}{ | c | c | c | c | c |  }
    \hline
   $\log(\varepsilon)$ & $m$ & $\rho/\max\{\Lambda_1^0,\Lambda_2^0\}$ & $R/\max\{\Lambda_1^0,\Lambda_2^0\}$& $\bar{t}\ (y)$\\ \hline \hline
    $-14.00$ & $22$ & $1.67\times 10^{-6}$ & $ 1.38\times 10^{-6}$ & $1.99\times 10^{10}$ \\ \hline
    $-12.25$ & $19$ & $1.22\times 10^{-6}$ & $ 1.13\times 10^{-6}$ & $3.04\times 10^{9}$  \\ \hline
    $-11.5$ & $18$ & $1.03 \times 10^{-6}$ & $8.09\times 10^{-7}$ & $1.04\times 10^{9}$ \\ \hline
  \end{tabular}
\captionof{table}{{\footnotesize From left to right: magnitude of the perturbation, number of iterative steps, real radius of the polydisk in the actions, radius of initial conditions in the actions, time of stability. Initial conditions in the actions are contained in an interval of radius $R$, whereas the initial values for the eccentricities are set to be those for Jupiter and Saturn at J2000. }}
\end{center}

\begin{figure}
\centering
\begin{subfigure}{.45\textwidth}
  \centering
  \includegraphics[width=1\linewidth]{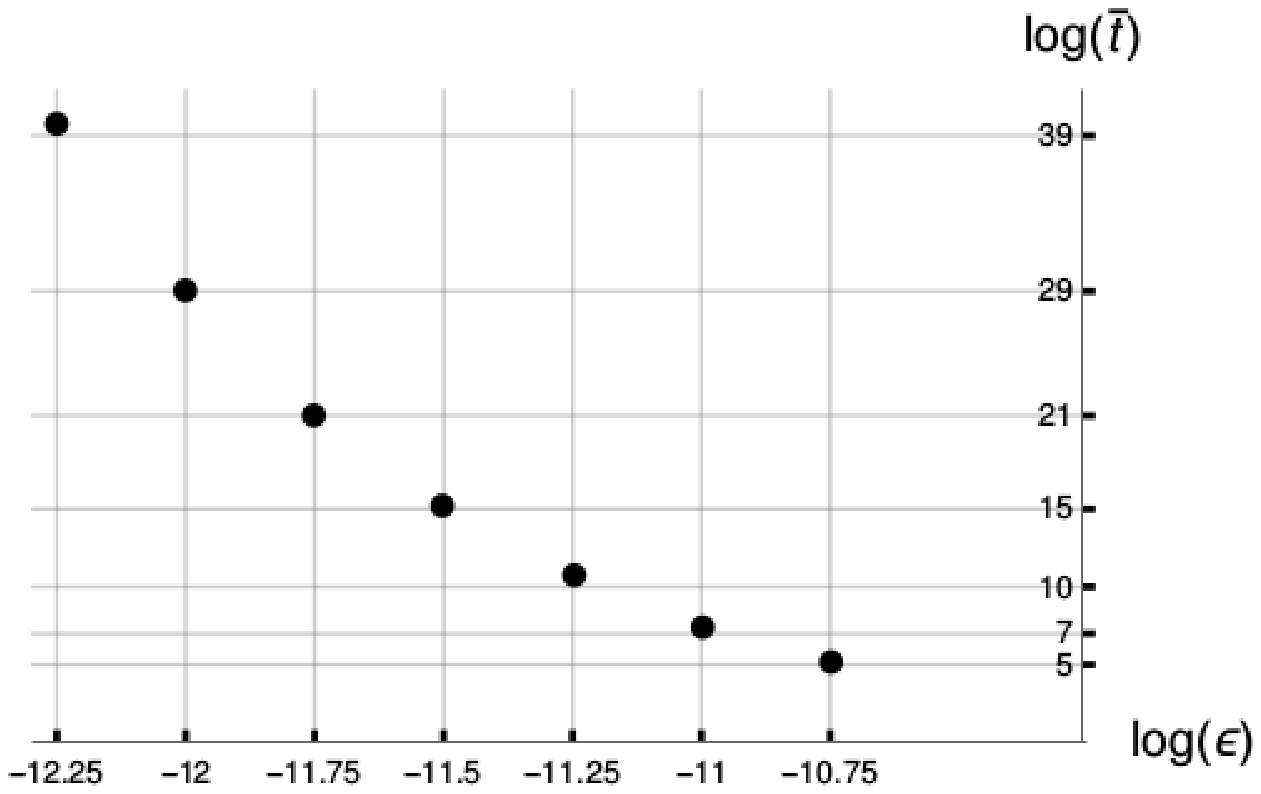}
  \label{fig:sub1}
\end{subfigure}%
\hspace{15pt}
\begin{subfigure}{.45\textwidth}
  \centering
  \includegraphics[width=1\linewidth]{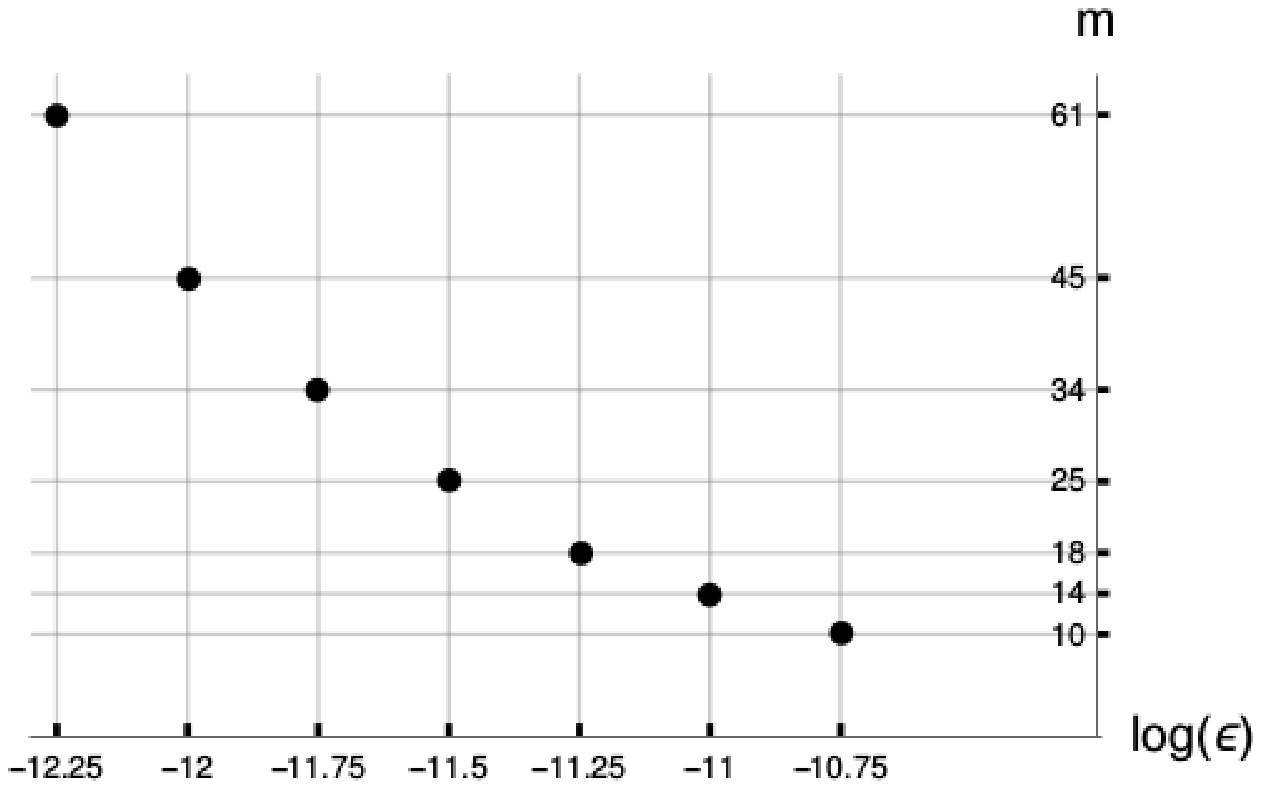}
  \label{fig:sub2}
\end{subfigure}
\begin{subfigure}{.45\textwidth}
  \centering
  \includegraphics[width=1\linewidth]{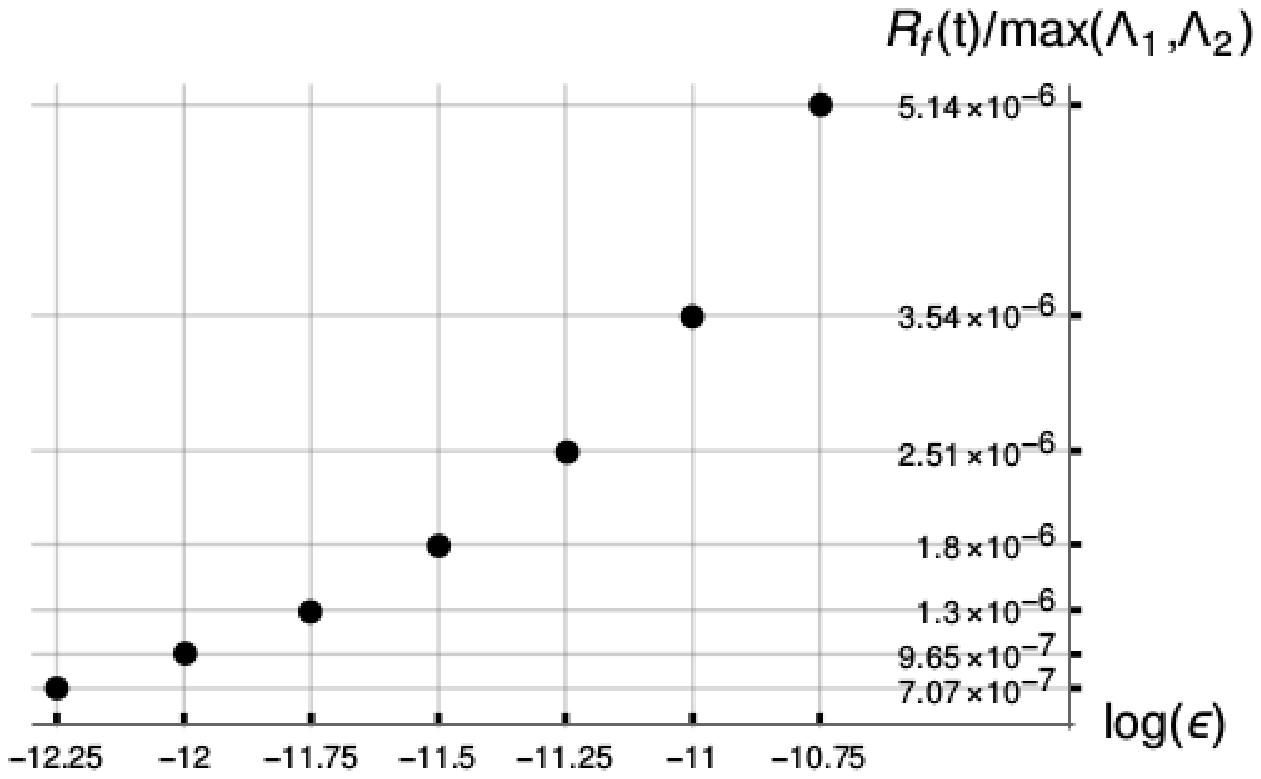}
  \label{fig:sub3}
\end{subfigure}%
\hspace{15pt}
\begin{subfigure}{.45\textwidth}
  \centering
  \includegraphics[width=1\linewidth]{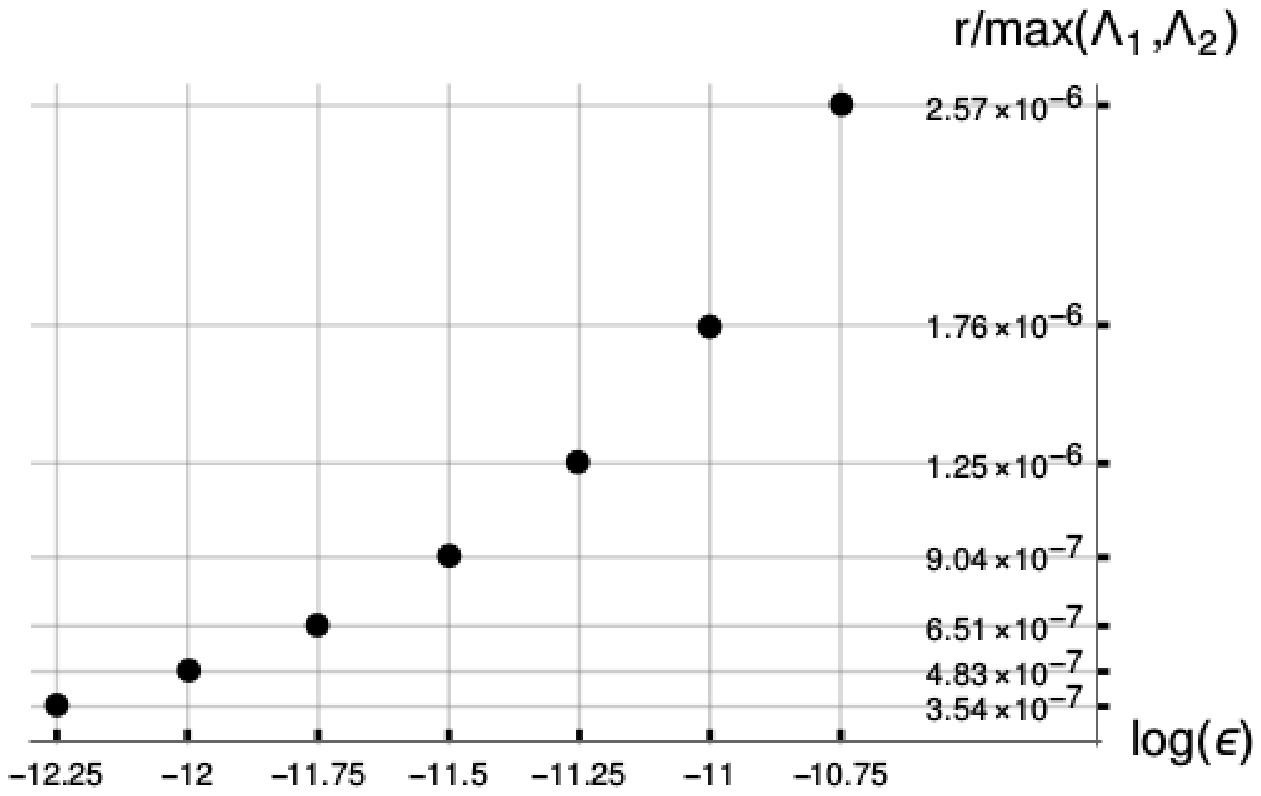}
  \label{fig:sub4}
\end{subfigure}%
\caption{{\footnotesize In clockwise sense starting from upper left: superlinear dependence of the maximal time of stability on the size of the perturbation; decrease of the best number of perturbative steps $m$; increase of the value of the analiticity width $r$ yielding the longest time of stability; increase of the radius of confinement of the action variables. Initial conditions are the same of Table 1. }}
\label{fig:test}
\end{figure}
\
\newline
Indeed, we notice that the best number of iterations $m$ decreases quite rapidly when $\varepsilon$ undergoes even small variations. This prevents one from obtaining a time of stability comparable with the timescale of the problem (which is the estimated age of the Solar System, i.e. about $5\times 10^9$ years) for higher values of $\varepsilon$ in the resonant regime. However, the results we obtained improve those achieved with the same techniques by other authors. Indeed, Niederman reached $\bar{t}\sim 4\times 10^9 $ years for $\varepsilon<10^{-13}$ in \cite{Niederman_1996}, whereas Castan obtained $\bar{t}\sim 1.3\times 10^{11} $ years for $\varepsilon<10^{-13}$ in \cite{Castan_2017}. In our case, since we made use of sharp methods based on vector field estimates, we were able to get good times of stability (i.e. greater or equal, say, than $1\times 10^{9}$ years) for values of $\varepsilon$ which are almost $100$ times greater than those in \cite{Niederman_1996} and in \cite{Castan_2017}, even though the theory is flawed, as we have just said, by the fast descrease of $m$ when $\varepsilon$ increases. This phenomenon, in turn, appears to be due to condition (\ref{hyp_it_1}) in lemma \ref{iterative},
$$
\frac{Tm\eo}{2}<1\ ,
$$
which insures that each iteration actually diminishes the magnitude of the non-resonant perturbation. By making use of the notations in paragraph \ref{analiticity}, one can equivalently rewrite it in the form
$$
\frac{Tm\varepsilon\nub{H_P}{4}}{2rs}<1\ .
$$
By looking at this expression, when considering increasing values for $\varepsilon$ one would be tempted to increase in turn $r$ or $s$ in order to compensate such growth and keep $m$ sufficiently high. Such strategy only works up to a certain point. Indeed, the constant $C_1(R)$ appearing in theorem \ref{stability} increases as $r^2$, but a huge value of $r$ amounts to enlarging the domain in which $\nub{H_P}{4}$ is estimated and, moreover, it entails a remarkable growth on the parameter $\delta$ associated with the remainder of order two for the unperturbed hamiltonian. In particular, the size of $\delta$ appears to be essential in this scheme, since it represents, roughly speaking, the {\itshape distance} to the resonance. Thus, increasing $r$ becomes helpless beyond a certain threshold. One may also be tempted to do the same thing with $s$ to keep the above inequality true. Unfortunately, this does not work at all since $s$ is the only analyticity width which is involved in the exponential stability (see expression (\ref{C1}) for $C_2$ in theorem \ref{stability} and take the definition of $\eo$ into account): even slight variations in its value lead to large deteriorations in the time of stability. Moreover, since the Fourier harmonics of $H_P$ diverge exponentially in the imaginary direction, a remarkable increase in $\nub{H_P}{4}$ is entailed when increasing $s$. A possible way to overcome such difficulty may be a sharper estimate on the size of the complex hamiltonian which does not make use of majorant series. More powerful techniques of perturbation theory may also be implemented, such as continuous averaging (see \cite{Treschev_Zubelevich_2010}).
\newline
When considering a non-zero radius $R$ of initial conditions in the action variables, we remark that, even in case a relatively large number of iterative steps $m$ is still available, results worsen if $R$ is too large and the system is thus too far from the resonant unperturbed torus. Such behaviour is due, once more, to the growth of the term $\delta$. In the sequel, we will see that this phenomenon arises dramatically when considering the same computations for the restricted, circular, planar problem.
\newline
Lastly, as we have already stressed, this study relies on rigorous estimates on the domain of analyticity for hamiltonian (\ref{ham_bello}) which are contained in \cite{Castan_2017}. In some sense, as we anticipated in paragraph \ref{motivation}, this opens an interesting discussion on the role of singularities in preventing Nekhoroshev stability. Actually, as the previous tables show, when considering increasing values for $\varepsilon$, one is also obliged to increase the size in the action variables of the domain of analyticity in order to get good times of stability. In our case, computations show that the magnitude of the complex hamiltonian grows significantly when considering a radius $r \sim 4\times 10^{-5}\times \max\{\Lambda_1^0,\Lambda_2^0\}$, so quite far from the region of the complex phase space that we are considering. Namely, the problem of having a low number $m$ of available perturbative steps for increasing values of $\varepsilon$ and the growth of $\delta$ appear well before singularities. However, as the same computations have shown, the latter may be an obstacle when dealing with non-sharp constants and when the initial estimates on functions and vector fields are rough. Indeed, in those cases one is obliged to choose smaller values for $\varepsilon$ and larger values for $r$ in order to get a good time of stability. In this light, singularities appear to be an essential difficulty when dealing with perturbation theory, at least when one considers the non-truncated model. It is interesting to notice that Tresch\"{e}v and Zubelevich pointed out the the importance of singularities in a different context when describing the continuous averaging method in \cite{Treschev_Zubelevich_2010}. 
\newline
In order to see what happened around different periodic tori, we also explored other resonances for the same masses, eccentricities and semi-major axis for the heavier planet: in all cases the first arising difficulty was the abrupt decrease in the optimal number of iterations $m$. Moreover, no significant improvement on the thresholds for $\varepsilon$ were reached. 
\newline
Finally, one should also remark that since $\beta\sim 1$ yields the best times of stability, the optimal choice for $u$ coincides in practice with the natural choice $u=\sqrt{rs}$. 
\subsection{The 3:1 resonance for the restricted problem}
As for the restricted case, we chose to study the $3:1$ resonance for a Sun-Jupiter-asteroid model (with smaller Jupiter's mass), as it corresponds to a region of phase space where the construction described in \cite{Celletti_Chierchia_2007} applies for suitable initial values of the eccentricity $e$. Indeed, for such model to hold, one needs the discarded harmonics to be smaller in value than those discarded in \cite{Celletti_Chierchia_2007}: this is precisely what we have checked preliminarily in our computations. Moreover, since in such case the perturbation is constructed by retaining only the most relevant harmonics from the complete perturbation, it is possible to compute a numerical averaging to higher orders in $\varepsilon$ in order to improve the thresholds for which theorem \ref{stability2} yields good times of stability. To achieve such goal, one can apply the near-to-identity transformations described in reference \cite{Celletti_Ferrara_1996}, where a different region in phase space for the same system is explored. Moreover, as we anticipated in paragraph \ref{motivation}, it is possible to have explicit expressions for the initial vector fields so that one can estimate their initial size without making use of the Cauchy inequalities. In particular, since we are working with analytic hamiltonians, the maximum modulus theorem (see \cite{Schiedemann_2005} for its statement and proof) insures that each function and each vector field component attains its maximum at the boundary of its domain. Therefore, our estimates were carried out by calculating the values of each function and each vector field component on a large number of randomly-chosen points belonging to the boundary of their domains and by taking their maximum. The chosen number of points was $10^6$ for each trial and multiple tests have been done to check that the estimates stayed stable for different random trials. Though not mathematically rigorous like those used in the planetary case, this method is an easy way to have a strong indication on initial estimates. If one wanted rigorous estimates (though the authors believe that they would not substantially differ from those obtained with the probabilistic draw described above) a possible solution avoiding Cauchy inequalities may involve the use of complex interval arithmetic (see e.g. \cite{Petkovic_1998}). Jupiter's eccentricity and semi-major axis are those calculated at J2000, we chose $e(0)\in[0,0.2]$ as the range of arbitrary initial values for the eccentricity of the massless body and we tried many different values for its semi-major axis in the neighborhood of the 3:1 resonance with Jupiter. As in the planetary case, the longest times of stability are obtained for an initial condition in the action $L$ corresponding exactly to the resonance. 
\newline
These, together with those obtained in \cite{Celletti_Ferrara_1996} in the non-resonant regime, are shown in the following table, where $N$ denotes the number of preliminary averagings to higher orders of the initial perturbation. We were only able to perform $N=1$ at most since more steps involved a huge increase in CPU time due to the randomly chosen boundary points involved in the initial estimates. However, even a single preliminary step gives a clear idea of how things work in the resonant regime we are considering. Indeed, the authors in \cite{Celletti_Ferrara_1996} deal with a high order completely non-resonant domain; nevertheless, we think it is interesting to compare the results obtained in the two cases, especially in terms of the thresholds on the perturbation, since a non-sharp version of Nekhoroshev theorem (originally stated in \cite{Poschel_1993}) was used in \cite{Celletti_Ferrara_1996}. 
\newline
\newline
\begin{center}
  \begin{tabular}{ | c | c | c | c | c | c | }
    \hline
    & {\footnotesize $N$} & {\footnotesize m} & {\footnotesize $\log(\varepsilon)$} & {\footnotesize $\bar{t}\ (y)$} & {\footnotesize $R_f(\bar{t})/\max\{\Lambda_1,\Lambda_2\}$} \\ \hline \hline
    {\footnotesize This work} & \footnotesize {$\ 0$} & {\footnotesize 8} & {\footnotesize $-8.75$} & {\footnotesize $2.13\times 10^{11}$} & {\footnotesize $2.87\times 10^{-5}$} \\ \hline
    {\footnotesize This work} & \footnotesize {$\ 1$} & {\footnotesize 8} & {\footnotesize $-7.00$} & {\footnotesize $1.20\times 10^{9}\ $} & {\footnotesize $3.25\times 10^{-6}$} \\ \hline
    {\footnotesize Celletti \& Ferrara (1996)} & {\footnotesize $\ 0$} & {\footnotesize -} & {\footnotesize $-13.00$} & {\footnotesize $1.13\times 10^{10}$} & {\footnotesize $4.47\times 10^{-6}$} \\ \hline
    {\footnotesize Celletti \& Ferrara (1996)} & {\footnotesize $\ 1$} & {\footnotesize -} & {\footnotesize $-8.00$} &  {\footnotesize $\ 1.13\times 10^{10}\ $} & {\footnotesize $2.00\times 10^{-7}$}  \\ \hline
  \end{tabular}
  \captionof{table}{{\footnotesize From left to right: number of preliminary averaging steps, number of iterative steps, size of the perturbation, time of stability, variation of the action variables.}}
\end{center}
\ 
\newline
As one can easily see, sharp estimates seem to play a role since the thresholds on the value of $\varepsilon$ yielding good times of stability are largely improved. It would be interesting to develop more powerful numerical tools in order to compare our sharp results with those in \cite{Celletti_Ferrara_1996} for higher values of $N$ ($N\leq 4$ in \cite{Celletti_Ferrara_1996}). However, we expect the confinement in the action variables to be less strong in our case, since we are in a low order resonant region. We also expect that a higher order of preliminary averaging $N$ would allow one to reach good thresholds on the allowed size of $\varepsilon$ and good times of stability.  
\newline
By any means, as far as we focus on the limits of the theory we deal with, our computations for the restricted problem show that the main issue is the growth with $\varepsilon$ of the bound $\delta$ on the remainder of order two in the developement of the unperturbed hamiltonian. As we showed when considering computations for the planetary case and as explicit estimates in theorem \ref{stability2} show, one is obliged to choose larger domains in the action variables when increasing the value of $\varepsilon$ in order to get a good time of stability. Therefore $\delta$ may become large, since averaging theory leaves the unperturbed hamiltonian untouched. This, in turn, prevents iterative lemma \ref{it_lem_2} from working properly (it may not dimish the size of the perturbation enough when $\delta$ is too big). One could attempt to hinder such growth by diminishing the analyticity width in the action variables, but this would only result in diminishing the time of stability since the costant $C_3$ in (\ref{C3}) increases as $r_L^2$. Our computations show that, for $N=0$, the growth of $\delta$ becomes preponderant when considering magnitudes for the perturbation such that $\varepsilon<10^{-10}$. Increasing the number $N$ of preliminary averaging steps seems thus the only possible way in order to get more realistic values for $\varepsilon$.
\appendix
\section{Proof of the estimates in lemma \ref{iterative}}\label{proof_est}
In this appendix, we give the proof of estimates (\ref{bound_chd_f}), (\ref{bound_cht_f}) and (\ref{bound_it_function}) in the statement of lemma \ref{iterative}.
\newline
We start by remarking that the hamiltonian vector field $\ch{r_1}$ of the remainder (\ref{reste}) in lemma (\ref{iterative}) is bounded by
\begin{align}\label{Mw}
\begin{split}
\nub{\ch{r_1}^j}{1-2\auno}\leq &
\integ{0}{1}{\sum_{k=1}^8\nub{\mathcal{M}^{jk}\left(\left[\ch{\phi_1},\ch{g_0+tf_0}\right]^k\circ\flphit\right)}{1-2\auno}}{t}\\
\leq &
\sum_{k=1}^8\nub{\mathcal{M}^{jk}}{1-2\auno}\nub{\left[\ch{\phi_1},\ch{g_0+tf_0}\right]^k}{1-\auno}\ .\\ 
\end{split}
\end{align}
Then we state the following 
\newtheorem{definition}[]{Definition}
\begin{definition}\label{def_bounded}
Let $\mathcal{A}$ be a $m\times n$ matrix whose entries $a_{jk}$, with $j\in\{1,...,m\}$ and $k\in\{1,...,n\}$, are complex-valued functions defined in a complex domain $\mathcal{E}$, i.e.
$$
a_{jk}:\mathbb{C}^l\supset\mathcal{E}\longrightarrow \mathbb{C}\ ,
$$
with $l$ a positive integer.
\newline
Let $\mathcal{B}$ be a $m\times n$ matrix with constant real entries.\newline
We say that $\mathcal{A}$ is bounded by $\mathcal{B}$ on $\mathcal{E}$, and we simply write $\mathcal{A}\leq\mathcal{B}$, iff
$$
|a_{jk}|_{\mathcal{E}}\leq b_{jk}\ \ \forall\ (j,k)\ \in\ \{1,...,m\}\times\{1,...,n\}\ .
$$
\end{definition}
\ 
\newline
With this definition, we can state that
\newtheorem{M_est}[]{Lemma} 
\begin{M_est}\label{M_est}
$\mathcal{M}$ is bounded on $\dom{1-2\auno}$ by a matrix 
$$\bar{\mathcal{M}}
\:=
\left(
\begin{matrix}
\bar{\mathcal{M}}_A & \bar{\mathcal{M}}_B\\
\bar{\mathcal{M}}_C & \bar{\mathcal{M}}_D\\
\end{matrix}
\right)
$$
whose blocks read
$$
\bar{\mathcal{M}}_A:= \left(
\begin{matrix}
\displaystyle\frac{T\eo}{2\auno}+1 & \displaystyle\frac{T\eo}{2\auno} & \displaystyle\sqrt{\frac{r}{s}}\frac{T\Eo}{2\auno\beta} & \displaystyle\sqrt{\frac{r}{s}}\frac{T\Eo}{2\auno\beta} \\
\ & \ & \ & \ \\
\displaystyle\frac{T\eo}{2\auno} & \displaystyle\frac{T\eo}{2\auno}+1 & \displaystyle\sqrt{\frac{r}{s}}\frac{T\Eo}{2\auno\beta} & \displaystyle\sqrt{\frac{r}{s}}\frac{T\Eo}{2\auno\beta} \\
\ & \ & \ & \ \\
\displaystyle\beta\sqrt{\frac{s}{r}}\frac{T\eo}{2\auno} & \displaystyle\beta\sqrt{\frac{s}{r}}\frac{T\eo}{2\auno} & \displaystyle\frac{T\Eo}{2\auno}+1 & \displaystyle\frac{T\Eo}{2\auno}\\
\ & \ & \ & \ \\
\displaystyle\beta\sqrt{\frac{s}{r}}\frac{T\eo}{2\auno} & \displaystyle\beta\sqrt{\frac{s}{r}}\frac{T\eo}{2\auno} & \displaystyle\frac{T\Eo}{2\auno} & \displaystyle\frac{T\Eo}{2\auno}+1
\end{matrix}
\right)\ ,
$$
$$
\bar{\mathcal{M}}_B:= \left(
\begin{matrix}
\displaystyle\frac{r}{s}\frac{T\eo}{2\auno} & \displaystyle\frac{r}{s}\frac{T\eo}{2\auno} & \displaystyle\sqrt{\frac{r}{s}}\frac{T\Eo}{2\auno\beta} & \displaystyle\sqrt{\frac{r}{s}}\frac{T\Eo}{2\auno\beta} \\
\ & \ & \ & \ \\
\displaystyle\frac{r}{s}\frac{T\eo}{2\auno} & \displaystyle\frac{r}{s}\frac{T\eo}{2\auno} & \displaystyle\sqrt{\frac{r}{s}}\frac{T\Eo}{2\auno\beta} & \displaystyle\sqrt{\frac{r}{s}}\frac{T\Eo}{2\auno\beta}  \\
\ & \ & \ & \ \\
\displaystyle\beta\sqrt{\frac{r}{s}} \frac{T\eo}{2\auno} & \displaystyle\beta\sqrt{\frac{r}{s}} \frac{T\eo}{2\auno} & \displaystyle\frac{T\Eo}{2\auno} & \displaystyle\frac{T\Eo}{2\auno} \\
\ & \ & \ & \ \\
\displaystyle\beta\sqrt{\frac{r}{s}} \frac{T\eo}{2\auno} & \displaystyle\beta\sqrt{\frac{r}{s}} \frac{T\eo}{2\auno} & \displaystyle\frac{T\Eo}{2\auno} & \displaystyle\frac{T\Eo}{2\auno} \\
\end{matrix}
\right)\ ,
$$
$$
\bar{\mathcal{M}}_C:= \left(
\begin{matrix}
\displaystyle\frac{s}{r}\frac{T\eo}{2\auno} & \displaystyle\frac{s}{r}\frac{T\eo}{2\auno} & \displaystyle\sqrt{\frac{s}{r}}\frac{T\Eo}{2\auno\beta} & \displaystyle\sqrt{\frac{s}{r}}\frac{T\Eo}{2\auno\beta} \\
\ & \ & \ & \ \\
\displaystyle\frac{s}{r}\frac{T\eo}{2\auno} & \displaystyle\frac{s}{r}\frac{T\eo}{2\auno} & \displaystyle\sqrt{\frac{s}{r}}\frac{T\Eo}{2\auno\beta} & \displaystyle\sqrt{\frac{s}{r}}\frac{T\Eo}{2\auno\beta}  \\
\ & \ & \ & \ \\
\displaystyle\beta \sqrt{\frac{s}{r}}\frac{T\eo}{2\auno} & \displaystyle\beta \sqrt{\frac{s}{r}}\frac{T\eo}{2\auno} & \displaystyle\frac{T\Eo}{2\auno} & \displaystyle\frac{T\Eo}{2\auno}\\
\ & \ & \ & \ \\
\displaystyle\beta \sqrt{\frac{s}{r}}\frac{T\eo}{2\auno} & \displaystyle\beta \sqrt{\frac{s}{r}}\frac{T\eo}{2\auno} & \displaystyle\frac{T\Eo}{2\auno} & \displaystyle\frac{T\Eo}{2\auno} \\
\end{matrix}
\right)\ ,
$$
$$
\bar{\mathcal{M}}_D:= \left(
\begin{matrix}
\displaystyle\frac{T\eo}{2\auno}+1 & \displaystyle\frac{T\eo}{2\auno} &  \displaystyle\sqrt{\frac{s}{r}}\frac{T\Eo}{2\auno\beta} & \displaystyle\sqrt{\frac{s}{r}}\frac{T\Eo}{2\auno\beta} \\
\ & \ & \ & \ \\
\displaystyle\frac{T\eo}{2\auno} & \displaystyle\frac{T\eo}{2\auno}+1 &  \displaystyle\sqrt{\frac{s}{r}}\frac{T\Eo}{2\auno\beta} & \displaystyle\sqrt{\frac{s}{r}}\frac{T\Eo}{2\auno\beta} \\
\ & \ & \ & \ \\
\displaystyle\beta \sqrt{\frac{r}{s}}\frac{T\eo}{2\auno} & \displaystyle\beta \sqrt{\frac{r}{s}}\frac{T\eo}{2\auno} & \displaystyle\frac{T\Eo}{2\auno}+1 & \displaystyle\frac{T\Eo}{2\auno}\\
\ & \ & \ & \ \\
\displaystyle\beta \sqrt{\frac{r}{s}}\frac{T\eo}{2\auno} & \displaystyle\beta \sqrt{\frac{r}{s}}\frac{T\eo}{2\auno} & \displaystyle\frac{T\Eo}{2\auno} & \displaystyle\frac{T\Eo}{2\auno}+1
\end{matrix}
\right)\ .
$$
\end{M_est}
\begin{proof}
We consider the Jacobian $D\flphit$ and we decompose it into $4\times 4$ matrix blocks
$$
D\flphit :=
\left(
\begin{matrix}
A & B\\
C & D \\
\end{matrix}
\right)\ ,
$$ 
\ 
\newline
Since the matrices $\mathcal{J}$ and $\mathcal{J}^{-1}=-\mathcal{J}$ act on the blocks of $D\flphit$ by mixing and transposing them, then $\mathcal{M}$ reads
$$
\mathcal{M}= \mathcal{J}(D\flphit)^\dag\mathcal{J}^{-1}=
\left(
\begin{matrix}
-D^\dag & \ \ B^\dag \\
\ \ C^\dag & -A^\dag \\
\end{matrix}
\right)\ .
$$
The proof of the statement follows by making use of the Cauchy inequalities for each entry of $\mathcal{M}$.
\end{proof}
\ 
\newline
Once $\mathcal{M}$ has been bounded, we must give an estimate to the Lie brackets appearing in expression (\ref{reste}). To do so, we use a result which is proven in \cite{Fasso_1990} and which we briefly recall in the sequel. 
\newline
\newline
Consider $\mathcal{E}$, an open and bounded domain of $\mathbb{R}^n$, and two vectors $\varsigma,\sigma\in \mathbb{R}^n$ with positive entries and such that for each component $\sigma_j<\varsigma_j,\ j\in\{1,...,n\}$. We define the complex polydisk $\mathcal{E}_{\varsigma}$ as
$$
\mathcal{E}_{\varsigma}:=
\{
z\in\mathbb{C}^n \text{ s.t. } |z_j-z_j^*|<\varsigma_j,\ z_j^*\in\mathcal{E}
\}
$$
and we have the following estimates on Lie and Poisson brackets:
\newtheorem{fasso}[]{Lemma}\label{Fasso1990}
\begin{fasso}
Let $X$ be a hamiltonian vector field analytic in $\mathcal{E}_{\varsigma}$, with an associated hamiltonian function $\mathcal{H}$.
\newline
Then:
\begin{enumerate}
\item For any function $f$ analytic in $\mathcal{B}_\eta$ one has 
\begin{equation}\label{fasso_f}
\nub{\{\mathcal{H},f\}}{\varsigma-\sigma}=\nub{L_X(f)}{\varsigma-\sigma}\leq \max_{j\in\{1,...,n\}}\left(\frac{\nub{X^j}{\varsigma-\sigma}}{\sigma_j}\right)\nub{f}{\varsigma}\ .
\end{equation}
\item For any vector field $Y$, analytic in $\mathcal{E}_{\varsigma}$, one has
\begin{equation}\label{fasso_ch}
\nub{\lie{X}{Y}^k}{\varsigma-\sigma}\leq \nub{X^k}{\varsigma}\max_{j\in\{1,...,n\}}\left(\frac{\nub{Y^j}{\varsigma}}{\sigma_j}\right)
+\nub{Y^k}{\varsigma}\max_{j\in\{1,...,n\}}\left(\frac{\nub{X^j}{\varsigma}}{\sigma_j}\right)\ .
\end{equation}
\end{enumerate}
\end{fasso}
\
\newline
As a straightforward consequence of this lemma we have the following
\newtheorem{lie_brackets}[]{Corollary}
\begin{lie_brackets}
The expression $\nub{\left[\ch{\phi_1},\ch{g_0+tf_0}\right]}{1-\auno}$ appearing in formula $(\ref{Mw})$ can be bounded by the quantity
\begin{equation}
\bar{w}:=
\frac{1}{\auno}
\left(
\begin{matrix}
\nub{\displaystyle\ch{\phi_1}^{I_j}}{1}\chi_0
+
\Theta_0\left(\nub{\ch{f_0}^{I_j}}{1}+\nub{\ch{g_0-\mathcal{G}}^{I_j}}{1}\right)
 \\
 \ \\
\nub{\displaystyle\ch{\phi_1}^{x_j}}{1}\chi_0
+
\Theta_0\left(\nub{\ch{f_0}^{x_j}}{1}+\nub{\ch{g_0-\mathcal{G}}^{x_j}}{1}\right)
 \\
 \ \\
  \nub{\displaystyle\ch{\phi_1}^{\vartheta_j}}{1}\chi_0
+
\Theta_0\left(\nub{\ch{f_0}^{\vartheta_j}}{1}+\nub{\ch{g_0-\mathcal{G}}^{\vartheta_j}}{1}+
\nub{\ch{\mathcal{G}}^{\vartheta_j}}{1}\right)
 \\
 \ \\
\nub{\displaystyle\ch{\phi_1}^{y_j}}{1}\chi_0
+
\Theta_0\left(\nub{\ch{f_0}^{y_j}}{1}+\nub{\ch{g_0-\mathcal{G}}^{y_j}}{1}\right)
 \\
\end{matrix}
\right)\ ,
\end{equation}
where $\Theta_0,\chi_0$ are defined in (\ref{chi0}).

\end{lie_brackets}
\
\newline
By plugging these estimates into expression (\ref{Mw}), one can find a bound on the hamiltonian vector field of the remainder which reads
\begin{align}
\begin{split}
\nub{\ch{r_1}^j}{1-2\auno}\leq
\sum_{k=1}^8\nub{\mathcal{M}^{jk}}{1-2\auno}\nub{\left[\ch{\phi_1},\ch{g_0+tf_0}\right]^k}{1-\auno}\leq \sum_{k=1}^8\bar{\mathcal{M}}^{jk}\bar{w}^k .\\ 
\end{split}
\end{align}
Estimates (\ref{bound_chd_f}) and (\ref{bound_cht_f}) follow immediately from the expression above when one takes into account expressions (\ref{def_g1}) and (\ref{def_f1}) as well as the definitions of the anisotropic norms.
\
\newline
In order to get an estimate on the remainder, on the other hand, we immediately remark that the latter can be bounded by expression
\begin{equation}\label{remainderexp}
\nub{r_1}{1-2\auno} = \nub{\integ{0}{1}{\{\phi_1,g_0+t f_0\}\circ\Lambda_{\phi_1}^t}{t}}{1-2\auno}\leq \nub{\{\phi_1,g_0\}}{1-\auno}+\nub{\{\phi_1,f_0\}}{1-\auno}\ .
\end{equation}
By applying formula (\ref{fasso_f}) to the two terms on the right side of this inequality and by taking the following estimate 
\begin{align*}
\nub{\phi_1}{1} := \nub{\frac{1}{T}\integ{0}{T}{t f_0\circ\Lambda_h^t}{t}}{1}
\leq \frac{T}{2}\nub{f_0}{1}
\end{align*}
into account one gets estimate (\ref{bound_it_function}).
\section{Proof of corollaries \ref{corol_it} and \ref{Corollary_nf}}\label{Prova_corollari}
\subsection{Proof of corollary \ref{corol_it}}
\begin{proof}
From the one-parameter group properties of the hamiltonian flow $\flphit$, one has
$$
\Lambda^{-t}_{\phi_1}\circ \flphit = id\ .
$$
Thanks to the linearity of the operator $L_{\phi_1}$ one can also write
$$
\Lambda_{\phi_1}^{-1}:= \exp{-L_{\phi_1}} = \exp{L_{-\phi_1}}:= \Lambda_{-\phi_1}^{1}\ ;
$$
consequently, the same estimates hold for $\flphit$ and $\Lambda^{-t}_{\phi_1}$.
\newline
Moreover, the following inclusion holds
\begin{align}
\begin{split}
\Lambda^{-1}_{\phi_1}\circ\Lambda^{1}_{\phi_1}(\dom{1-2\auno})=\dom{1-2\auno}\subset \Lambda^{-1}_{\phi_1}\left(\dom{1-\auno}\right)\ ,
\end{split}
\end{align}
so that finally we can define
\begin{align}
\begin{split}
&\Phi_1^{-1}:\dom{1-\auno}\longrightarrow\dom{1}\\ 
\ \\
&(I,\vartheta,x,y)\longmapsto \Lambda_{\phi_1}^{-1}(I,\vartheta,x,y)\ ,
\end{split}
\end{align}
and we can insure that estimates (\ref{size_inv}) hold.
\end{proof}
\subsection{Proof of corollary \ref{Corollary_nf}}
\begin{proof}
Consider $\Psi_m^{-1}:= \Phi_m^{-1}\circ ... \circ \Phi_1^{-1}$; one has
\begin{align*}
& \ndb{\Psi_m^{-1}-id}{1-\frac{T\eo}{2}\frac{1-q_1^m}{1-q_1},1-\frac{T\Eo}{2}\frac{1-q_2^m}{1-q_2}}
\leq  \\
& \ndb{\Phi_m^{-1}\circ ... \circ \Phi_1^{-1}-\Phi_{m-1}^{-1}\circ ...\circ \Phi_1^{-1}}{1-\frac{T\eo}{2}\frac{1-q_1^m}{1-q_1},1-\frac{T\Eo}{2}\frac{1-q_2^m}{1-q_2}}\\
\ \\
+ & \ndb{\Phi_{m-1}^{-1}\circ ... \circ \Phi_1^{-1}-\Phi_{m-2}^{-1}\circ ...\circ \Phi_1^{-1}}{1-\frac{T\eo}{2}\frac{1-q_1^m}{1-q_1},1-\frac{T\Eo}{2}\frac{1-q_2^m}{1-q_2}}\\
\ \\
+ & ...
+ \ndb{\Phi_1^{-1}-id}{1-\frac{T\eo}{2}\frac{1-q_1^m}{1-q_1},1-\frac{T\Eo}{2}\frac{1-q_2^m}{1-q_2}}\\
\ \\
\leq & \sum_{j=0}^{m-1}\frac{T\eta_j}{2} = \frac{T\eta_0}{2}\sum_{j=0}^{m-1} q_1^j
\leq \frac{T\eo}{2}\frac{1-q_1^m}{1-q_1}
\end{align*}
where we have made use of corollary (\ref{corol_it}) for each function $\Phi_j^{-1}$.
Similar reasonings yields
$$
\ntb{\Psi_m^{-1}-id}{1-\frac{T\eo}{2}\frac{1-q_1^m}{1-q_1},1-\frac{T\Eo}{2}\frac{1-q_2^m}{1-q_2}}\leq \frac{T\Eo}{2}\frac{1-q_2^m}{1-q_2}\ .
$$
On the other hand, it is straightforward to see that
$$
\dom{1-\frac{T\eo}{2}\frac{1-q_1^m}{1-q_1},1-\frac{T\Eo}{2}\frac{1-q_2^m}{1-q_2}} \subset \Psi_m\left(\dom{1}\right)\ ,
$$
so that $\Psi_m^{-1}$ is well defined.
\end{proof}
\section{Normal form for the restricted problem}\label{lem_res}
\newtheorem{res_nf_2}[ac_stability]{Lemma (normal form lemma)}\label{nf_lemma_res}
\begin{res_nf_2}
With the definitions above, suppose that there exist a positive integer $m$ and five real numbers $p,q_j,\text{with }j\in\{L,G,l,g\}, $ such that 
\begin{align}
2\upsilon_0^j(m)<q_j\ ,\ \ \ 
2\zeta_0(m)< p\ .
\end{align}
Then there exist a symplectic transformation $\Psi_m$, analytic and real-valued for any real argument
$$
\Psi_m: \dom{1}\longrightarrow\dom{1+\frac{T\eo^j}{2}\frac{1-q_j^m}{1-q_j}}\ ,
$$
whose size is 
\begin{equation}\label{size_nf2}
\ndb{\left(\Psi_m-id\right)^j}{1}\leq \frac{T\eo^j}{2}\frac{1-q_j^m}{1-q_j}\ ,\ \ \ \ j\in\{L,G,l,g\}
\end{equation}
such that
$$
H_m:= H_0\circ\Psi_m=h+g_m+f_m\ ,
$$
where $\{h,g_m\}=0$ 
and $ \langle f_m\rangle_h=0$.
\newline
Furthermore, one has the following estimates $(j\in\{L,G,l,g\})$:
\begin{align}\label{ch_nf2}
\begin{split}
\ndb{\ch{f_m}^j}{1}\leq &\ q_j^m\eo^j\\
\ndb{\ch{g_m}^j-\ch{\mathcal{G}}^j}{1}\leq &\ \go^j+\frac{q_j}{2}\frac{1-q_j^m}{1-q_j}\eo^j\\
\nub{f_m}{1}\leq &\ p^m\nub{f_0}{3}\\
\nub{g_m-\mathcal{G}}{1}\leq &\ \nub{g_0-\mathcal{G}}{3}+\frac{p}{2}\frac{1-p^{m}}{1-p}\nub{f_0}{3}\ .\\ 
\end{split}
\end{align}
\end{res_nf_2}
\ 
\newline
As in section \ref{stab}, such lemma can be demonstrated by iterating $m$ times the following
\newtheorem{it_lem_2}[ac_stability]{Lemma (iterative lemma)}
\begin{it_lem_2}\label{it_lem_2}
Assume the construction of section \ref{ham_fr_r} and suppose that for a real number $\auno \in \left(0,1\right)$ one has 
\begin{equation}\label{hyp_it_2}
\frac{T}{2\auno}\max\{\eo^L,\eo^G,\eo^l,\eo^g\}< 1.\ 
\end{equation}
Then there exist a symplectic analytical transformation $\Phi_1$ of generating function $\phi_1$
$$
\Phi_1: \dom{3-2\auno}\longrightarrow \dom{3-\auno}\ ,
$$
which is real valued for any real argument, and whose size is 
\begin{equation}\label{size_it2}
||(\Phi_1-id)^j||_{3-2\auno}\leq \frac{T\eo^j}{2}\ \ \ \ j\in\{L,G,l,g\},
\end{equation}
which takes the hamiltonian into the following form:
$$
H_1:= H_0\circ\Phi_1= h+g_1+f_1\ ,
$$
where $\{h,g_1\}=0$ and $\langle f_1\rangle_h =0$. 
\newline
\newline
Furthermore, one has the following estimates on functions and vector fields
\begin{align}\label{est2}
\begin{split}
\nubfunos \leq &\ 2 \zeta_0\left(\frac{1}{\auno}\right) |f_0|_3,\ \ \ \nubgunos \leq \zeta_0\left(\frac{1}{\auno}\right)|f_0|_3+\nub{g_0-\mathcal{G}}{3}\ ,\\
\ndbfunos \leq &\ 2\upsilon_0^j\left(\frac{1}{\auno}\right)\eo^j,\ \ \ \ndbgunos \leq \upsilon_0^j\left(\frac{1}{\auno}\right)\eo^j + \go^j\ ,
\end{split}
\end{align}
where $j\in\{L,G,l,g\}$.
\end{it_lem_2}

The normal form lemma and the iterative are proven exactly as lemmas \ref{nf_lemma} and (\ref{iterative}), so we omit their demonstrations.
Moreover, a corollary on the existence of the inverse transformation for the normal form holds also in this case. Its statement and proof are exactly the same of corollary \ref{Corollary_nf}, so we omit them as well.

\section*{Acknowledgements}
We are indebted to Dr. Thibaut Castan for providing the codes to estimate the size of the three-body planetary problem perturbation on complex domains and for many interesting discussions on the subject. We would also like to thank Prof. Philippe Robutel for his useful remarks and Prof. Francesco Fass\`{o} for his advices when dealing with Mathematica.

\bibliographystyle{plainnat}
\bibliography{Article}

\end{document}